\documentclass[12pt, draftclsnofoot, onecolumn]{IEEEtran}
\usepackage{amsmath,amssymb,amsfonts,amsthm,mathrsfs}
\usepackage{graphicx}
\usepackage{balance}
\usepackage{graphics}
\usepackage{graphicx}
\usepackage{epsfig}
\usepackage{algorithm}              
\usepackage{algorithmic}             
\usepackage{multirow}                
\usepackage{amsmath}
\usepackage{xcolor}
\usepackage{graphicx}
\usepackage{subfigure}
\usepackage{stfloats}
\usepackage{bm}
\usepackage{bbm}
\usepackage{fancyhdr}

\usepackage{amsthm}
\newtheorem{theorem}{Theorem}
\newtheorem{lemma}{Lemma}
\theoremstyle{definition}
\newtheorem{definition}{Definition}
\theoremstyle{remark}

\theoremstyle{corollary}

\begin{document}\pagestyle{empty}
    
    \IEEEoverridecommandlockouts
    
    \title{On the Effective Throughput of Coded Caching: A Game Theoretic Perspective}

    \author{
        Yawei~Lu,~\IEEEmembership{Student Member,~IEEE},~Wei~Chen,~\IEEEmembership{Senior Member,~IEEE},\\~and~H. Vincent Poor,~\IEEEmembership{Fellow,~IEEE}
        
        \thanks{Yawei Lu and Wei Chen are with the Department of Electronic Engineering, Tsinghua University, Beijing, 100084, China, e-mail: lyw15@mails.tsinghua.edu.cn, wchen@tsinghua.edu.cn. Yawei Lu is also with the Department of Electrical Engineering, Princeton University, New Jersey, 08544, USA, e-mail: yawil@princeton.edu.}

        \thanks{H. Vincent Poor is with the Department of Electrical Engineering, Princeton University, New Jersey, 08544, USA, e-mail: poor@princeton.edu.}
        
    \thanks{ This research was supported in part by the U.S. National Science Foundation under Grants CCF-0939370 and CCF-1513915, the Chinese national program for special support for eminent professionals (10,000-talent program), Beijing Natural Science Foundation (4191001), and the National Natural Science Foundation of China under Grant No. 61671269. }

    }
    
    \maketitle
    \pagestyle{fancy}  
    \thispagestyle{fancy} 
    \lhead{} 
    \chead{} 
    \rhead{\thepage} 
    \lfoot{} 
    \cfoot{} 
    \rfoot{} 
    \renewcommand{\headrulewidth}{0pt} 
    \renewcommand{\footrulewidth}{0pt} 
    
    \maketitle\thispagestyle{empty}

\begin{abstract}
Recently coded caching has emerged as a promising means to handle continuously increasing wireless traffic. However, coded caching requires users to cooperate in order to minimize the overall  transmission rate.  How users with heterogeneous  preferences cooperate in coded caching and how to calculate the resulting caching gains are still  open problems. In this paper, a two-phase cache-aided network is investigated, in which users with heterogeneous preferences  are served by a base station through a shared link.  Effective throughput is considered as a performance metric. It is proved that the achievable domain of effective throughputs is a  convex set and can be characterized by its boundary. A special type of caching schemes, named uncoded placement absolutely-fair (UPAF) caching, is studied. For the two-user case, games are formulated to allocate effective throughput gains for the two users. For the general multiuser case, a  UPAF policy is proposed to organize user cooperation. It is shown that users with more concentrated preferences can obtain  higher effective throughputs.
\end{abstract}


\IEEEpeerreviewmaketitle

\section{Introduction}
    Thanks to the popularization of mobile devices and the rise of multimedia applications, wireless traffic has undergone dramatic growth in recent years. In the era of 5G, various wireless communication technologies have been developed to improve network capacity and  handle the ever-increasing traffic, including polar codes,  massive multi-input and multi-output, millimeter wave,  etc. Notice that a significant feature of current wireless traffic is that content items are generated much earlier than they are requested. Caching has emerged as an inexpensive and powerful solution to cope with increasing wireless traffic by pre-storing popular content items in user terminals before users request them \cite{5g}.

      By shifting traffic at peak hours to idle hours,  caching can exploit  idle communication resources.  Wireless networks benefit from caching in many aspects. Since content items can be provided for users from  closer nodes rather than the core network,  caching contributes to lower latency communications \cite{latency1}.  From an information theoretic perspective, the fundamental latency-cache tradeoff was revealed in \cite{latency2} by introducing a performance metric termed normalized delivery time. Considerable attention has also been paid to improving energy efficiency through caching in wireless networks \cite{energy1}-\cite{energy3}. In \cite{energy1}, an in-memory storage method was proposed to reduce energy consumption in edge caching. The mismatch between the randomly arrived  energy and content requests was eliminated by caching in energy harvesting networks \cite{energy2}. The energy efficiency gains from caching in interference channels were revealed in \cite{energy3}.   In millimeter wave communications, caching helps to enhance mobility support and reduce handover failures \cite{mobility1}.     A cost-optimal caching scheme was proposed for device-to-device (D2D) networks with user mobility in \cite{mobility2}. The effective throughput resulting from caching, which indicates the reduction in the real-time traffic, was investigated in \cite{eff1}-\cite{eff4}. 
     
    Recently, a caching method termed coded caching  has  attracted considerable attention \cite{fund1}-\cite{privacy}, because it can provide global caching gains.  The seminal works \cite{fund1} and \cite{fund2} studied a two-phase (placement phase and delivery phase) system and first revealed the fundamental memory-rate tradeoff for coded caching. Thereafter, the exact memory-rate tradeoff for coded caching with uncoded prefetching was characterized in \cite{fund3}.   Coded caching was further extended to  D2D communications in \cite{d2d}. It was shown that coded caching can create multicasting opportunities for overloaded multi-input single-output channels \cite{multiante}. The work \cite{multiserver} provided a coded caching scheme for multiple-server networks. An order-optimal coded caching scheme was proposed to meet users' privacy requirement in \cite{privacy}.

     In present multimedia applications, a few  popular content items typically account for the majority of network traffic \cite{popu1}.     User preferences on these popular content items have been shown to play a critical role in designing caching schemes \cite{critical}.   Various caching schemes were developed to improve network capability by taking advantage of content popularity \cite{popu2}-\cite{uneven4}. The work \cite{popu2} proposed femtocaching to reduce downloading delay according to network topology and popularity distributions. A file grouping scheme was developed to handle the uneven preferences in cache-aided networks  \cite{uneven1}. For content items with discrete popularities, \cite{uneven2} proposed a near-optimal coded caching scheme to reduce   network load. The memory-rate tradeoff for coded caching under uneven popularities was   revealed in \cite{uneven3} and \cite{uneven4}. Game theoretic techniques were also adopted to design caching and pricing strategies \cite{game1}-\cite{game2}.


In previous studies, the focus has been on improving network performance for the sake of servers. For example, \cite{fund1}-\cite{fund2} devoted to  reduce the load that the server sustains and \cite{energy} tried to lower the overall energy consumption. Less attention has been paid to the revenue that an individual user can obtain from caching. Thanks to the progress of big data technologies, user preferences can be analyzed based on private browsing history and social relationships  \cite{historical}-\cite{social}.  Existing studies have mostly ignored users' personal differences by assuming that all the users have the same probability of requesting a  content  item \cite{popu2}-\cite{game2}.  However, user preferences have a significant impact on caching gains.    If a user is interested  in only a few content items, i.e., the user requests are deterministic or near-deterministic, this user's buffer can be filled with these content items. This user can find the desired content items directly when it issues requests. As a result, this user obtains a high effective throughput from caching. If a user is interested in many content items, the above mechanism does not work due to the buffer size constraint. It can cooperate with other users by coded caching to split the transmission cost. Few works have focused on how to calculate each user's caching gain in this case. 

In this paper, we investigate effective throughput from caching for users with heterogeneous preferences.  More specifically, multiple users are served by a base station through a shared link and user preferences are characterized by a probability measure.  Similar to the model considered in \cite{fund1} and \cite{fund2}, the network works in two phases,  a placement phase and a delivery phase. In the placement phase, network load is light and  user buffers are filled through idle spectra. In the delivery phase, user requests are revealed and the base station transmits messages to help the users recover the desired content items. Effective throughput is used as a performance metric, which indicates the reduction in the transmission cost in the delivery phase. 
Each user's effective throughput is calculated individually.
To characterize the whole achievable domain of  effective throughputs, we need to investigate all the feasible placement and delivery policies, which is however of  prohibitive  complexity. Upon that, we prove the convexity of the achievable domain and focus on a special type of policies, termed uncoded placement absolutely-fair (UPAF) policies.  The achievable domain under UPAF policies is shown to be a polygon in the two-user case. 

The higher the effective throughput a user obtains, the lower the real-time transmission cost this user  affords. If the users are selfish and each user wishes to maximize its own effective throughput, the users form a game relationship. Based on the analysis on achievable domain in the two-user case, a noncooperative game is formulated to investigate the effective throughput equilibrium between the two users. In addition, a cooperative game is studied to allocate the revenue of cooperation,  which helps in designing pricing policies. Suffering from the hardness in finding a Nash equilibrium in noncooperative games, a low-complexity numerical algorithm is proposed to give a reasonable revenue allocation for the two users. An algorithm is also presented to organize  user cooperation in the general multiuser case. 

     The rest of the paper is organized as follows. Section II presents our system model and the formal definition of placement and delivery policies. Section III proves the convexity of the achievable domain of  effective throughputs and investigates the achievable domain for the two-user case in detail. Games among two users are studied in Section IV. Section V presents a cooperation scheme for the general multiuser case. Simulation results are given in Section VI. Finally, Section VII concludes this paper and suggests some directions for future research. 

    
\section{Problem Setting}

We introduce the system model in Subsection \ref{system_model}, provide the definition of effective throughput in Subsection \ref{def_eff}, and then present an example to illustrate the research motivation in Subsection \ref{motivation}.

\subsection{System Model}\label{system_model}
Consider a Base Station (BS) connected with $K$ users  through a shared error-free link. 
The BS has access to a database of $N$ content items, denoted by $W_1,\ldots,W_N$. Assume that all the content items have an identical size of $F$ bits. Let $\Omega = \{W_n:n\in [N]\}$ denote the collection of bits of the content items.\footnote{For a positive integer $A$, $[A]$ denotes the set $\{1,2,...,A\}$}  User $k$ is equipped with a buffer of $b_kF$ bits, or equivalently, $b_k$ content items.   We refer to  $\bm{B}=[b_1,...,b_K]$ as  buffer size vector. Let $d_{k,n}$ indicate whether  user $k$  asks for content item $W_n$, i.e., $d_{k,n}=1$ if user $k$  asks for $W_n$ and $d_{k,n}=0$ otherwise. We refer to $\bm{D}=(d_{k,n})_{K\times N}$ as the demand matrix, which is a random matrix with support set $\{0,1\}^{K\times N}$.  Then user preferences can be characterized by a probability measure on $\bm{D}$, denoted as $\mathcal{P}$. In practice, user preferences can be analyzed based on private browsing history and social relationships  \cite{historical}-\cite{social}. The probability that user $k$ requests $W_n$ is given by $p_{k,n}=\mathcal{P}(\{\bm{D}:d_{k,n}=1\})$. The matrix $\bm{P}=(p_{k,n})_{K\times N}$ is referred to as the user preference matrix. It should be noted that we do not assume that each user requests only one content item. In other words, the  summations $\sum_{n\in[N]}d_{k,n}$ and $\sum_{n\in[N]}p_{k,n}$ may not be 1.

This cache-aided network operates in two phases, namely a placement phase and a delivery phase. In the placement phase,  user requests are not specific. The users prefetch data from the BS and cache them in their buffers with the knowledge of  the probability measure $\mathcal{P}$.  In the delivery phase, the users issue requests for the content items. The demand matrix $\bm{D}$ reduces to a deterministic matrix $\bm{\hat{D}}=(\hat{d}_{k,n})_{K\times N}$.
 Local buffers provide useful information in recovering the requested content items. The users probably also need to turn to the BS in order to recover all the requested content items. The network described above is referred to as $(\bm{B},N,\mathcal{P})$-Caching.

\subsection{Formal Problem Statement}\label{def_eff}
We provide a formal description of  placement and delivery policies for $(\bm{B},N,\mathcal{P})$-Caching. 

\begin{definition}
    A placement and delivery policy $\pi = \left((\pi^p_k)_{k\in [K]},(\pi^g_{\mathcal{U}})_{\mathcal{U}\subseteq [K]},(\pi^r_k)_{k\in [K]}\right)$ consists of the following three types of functions.     \\
    i) Content prefetching function $(\pi^p_k)_{k\in [K]}$:  For each user $k$, $\pi_k^p$  determines the bits prefetched from the BS in the placement phase and thus gives this user's buffer state in the beginning of the delivery phase. Specifically, we have 
    \begin{equation}
        \pi_k ^p: (W_n)_{n\in [N]}, \mathcal{P} \rightarrow C_k.
    \end{equation}
  Due to the buffer size constraint, $C_k$ should satisfy $|C_k|\le b_kF$.\footnote{For a set $S$, $|S|$ denotes its cardinality.}\\
  ii) Message generation function $(\pi^g_{\mathcal{U}})_{\mathcal{U}\subseteq [K]}$: In the delivery phase, user requests are revealed and hence $\bm{\hat{D}}$ is known. For each subset $\mathcal{U}$ of $[K]$, $\pi^g_{\mathcal{U}}$ generates a message $M_{\mathcal{U}}$ according to $C_k$ and $\bm{\hat{D}}$, i.e., 
  \begin{equation}
  \pi^g_{\mathcal{U}} : (W_n)_{n\in [N]}, C_k,\bm{\hat{D}} \rightarrow M_{\mathcal{U}}.
  \end{equation}
The message $M_{\mathcal{U}}$ is generated to help users in $\mathcal{U}$  recover the requested content items.\\
iii) Content recovering function  $(\pi^r_k)_{k\in [K]}$: After receiving the transmitted messages, each user attempts to recover the requested content items by $\pi^r_k$, i.e., 
\begin{equation}
\pi_k^r : C_k, (M_{\mathcal{U}})_{\mathcal{U}:k\in \mathcal{U}} \rightarrow \{\hat{W}_n^k: n\in \hat{D}_k\},
\end{equation}
where $\hat{D}_k = \{n:\hat{d}_{k,n}=1\}$ represents the set of content items requested by user $k$.\footnote{$\mathcal{U}:k\in \mathcal{U}$ traverses all the subsets of $[K]$ that contain $k$.}
$\hat{W}_n^k$ stands for the estimated $W_n$.
\end{definition}

 For a placement and delivery policy $\pi$, the error probability is defined as
\begin{equation}
    \max_{k\in [K]} \max_{n\in [N]} \text{Pr}\{\hat{W}_n^k\neq W_n|d_{k,n}=1\},
\end{equation}
where $\text{Pr}\{\cdot|\cdot\}$ denotes the conditional probability.
Given a $(\bm{B},N,\mathcal{P})$-Caching, there are numerous  placement and delivery policies that can satisfy the user requests.
A traditional one is just to ignore user buffers and transmit the requested content items in the delivery phase directly. Caching and multicasting enable us to satisfy user requests in a more effective manner. For a policy $\pi$, we define the effective throughput
of user $k$ as 
\begin{equation}\label{effective_throughput}
\begin{split}
    R_k =&\frac{1}{F} \mathbb{E}_{\bm{D}}\left(\sum_{n\in [N]}Fd_{k,n}-\sum_{\mathcal{U}:k\in \mathcal{U}}\frac{|M_{\mathcal{U}}|}{|\mathcal{U}|}\right)\\
    =& \sum_{n\in [N]}p_{k,n}-\mathbb{E}_{\bm{D}}\left(\sum_{\mathcal{U}:k\in \mathcal{U}}\frac{|M_{\mathcal{U}}|}{|\mathcal{U}|F}\right),
\end{split}
\end{equation}
where $\mathbb{E}_{\bm{D}}(\cdot)$ denotes the mathematical expectation with respect to the random matrix $\bm{D}$. The summation $\sum_{n\in [N]}Fd_{k,n}$ stands for the number of bits transmitted to user $k$ in the delivery phase without caching and multicasting. Because $M_\mathcal{U}$ is transmitted to $|\mathcal{U}|$ users, each user incurs $\frac{1}{|\mathcal{U}|}$ cost of the transmission. Then, $\sum_{\mathcal{U}:k\in \mathcal{U}}\frac{|M_{\mathcal{U}}|}{|\mathcal{U}|F}$ indicates the total  transmission cost afforded by user $k$.  
In addition, we normalize the transmission cost by $\frac{1}{F}$. 
 Thus, the effective throughput defined in Eq. (\ref{effective_throughput}) represents the reduction in the transmission cost of user $k$.  
 
 The summation of all users' effective throughputs has
 \begin{equation}\label{total_eff}
 \begin{split}
 \sum_{k\in[K]}R_k =\sum_{k\in[K]}\sum_{n\in [N]}p_{k,n}-\frac{1}{F}\sum_{\mathcal{U}\subseteq [K]}\mathbb{E}_{\bm{D}}\left({|M_{\mathcal{U}}|}\right).
 \end{split}
 \end{equation}
 Note that $\sum_{k\in[K]}\sum_{n\in [N]}p_{k,n}$ and $\sum_{\mathcal{U}\subseteq [K]}\mathbb{E}_{\bm{D}}\left({|M_{\mathcal{U}}|}\right)$ represent the expected number of content items requested by the users and the total number of bits transmitted by the BS, respectively. The right-hand side of Eq. (\ref{total_eff}) indicates the reduction in the number of the bits transmitted by the BS due to caching and multicasting (normalized by the content size $F$). Thus, $R_k$ indicates the revenue of user $k$ from caching. As  caching schemes studied in previous literatures like \cite{fund1}-\cite{d2d}, the proposed policies  may incur a high signaling overhead. This overhead can be partly eliminated by large content items.

Users' effective throughputs vary with policies. All the possible values of effective throughputs  form an achievable domain.

\begin{definition}
    
    The  vector $(R_k)_{k\in [K]}$ is \emph{achievable} if for every $\varepsilon >0$ and every sufficiently large $F$  there exists a policy $\pi = \pi (\varepsilon, F)$ that achieves $(R_k)_{k\in [K]}$ with error probability lower than $\varepsilon$. 
   For a $(\bm{B},N,\mathcal{P})$-Caching, the \emph{achievable domain} of effective throughputs is defined as
    \begin{equation}
    \begin{split}
   \mathcal{ R}(\bm{B},N,\mathcal{P})= \{(R_k )_{k\in [K]}  : (R_k )_{k\in [K]} \mbox{ is achievable} \}.
    \end{split}     
    \end{equation}
\end{definition}

If no confusion arises, we  simply denote $ \mathcal{ R}(\bm{B},N,\mathcal{P})$ by $ \mathcal{R}$. The achievable domain $ \mathcal{R}$ depends on  the buffer size vector, the number of content items, as well as user preferences. 
  
  \subsection{A Motivating Example}\label{motivation}

  In this subsection, we present a demo to illustrate the research motivation and the impact of user preferences on effective throughputs. As shown in Fig. \ref{preferencecost}, two users are interested in two content items, denoted as $A$ and $B$. Each user at most caches one content item. User  1 requests the two content items with probability 99\% and 1\%, respectively. User 2 requests the two content items with an identical probability, 50\%. In this demo, we assume each user requests only one content item. The coded caching scheme suggests to divide each content item into two portions and then the user requests can always be satisfied by transmitting a coded packet of size 0.5 \cite{fund1}. Then both the two users achieve an effective throughput 0.75. The coded caching scheme requires the two users cooperate in both the two phases. Let us consider a noncooperation scheme that each user caches the most popular content items. The two users split the transmission costs if the delivered data are useful for both the two users.  It is seen that user 1 achieves a much higher effective throughput in the noncooperation scheme than it does in coded caching.  As a result, user preferences have a significant impact on caching gains and should be taken into account  in cache-aided networks. 
  
    \begin{figure} 
      \centering
      \includegraphics[width=10.3cm]{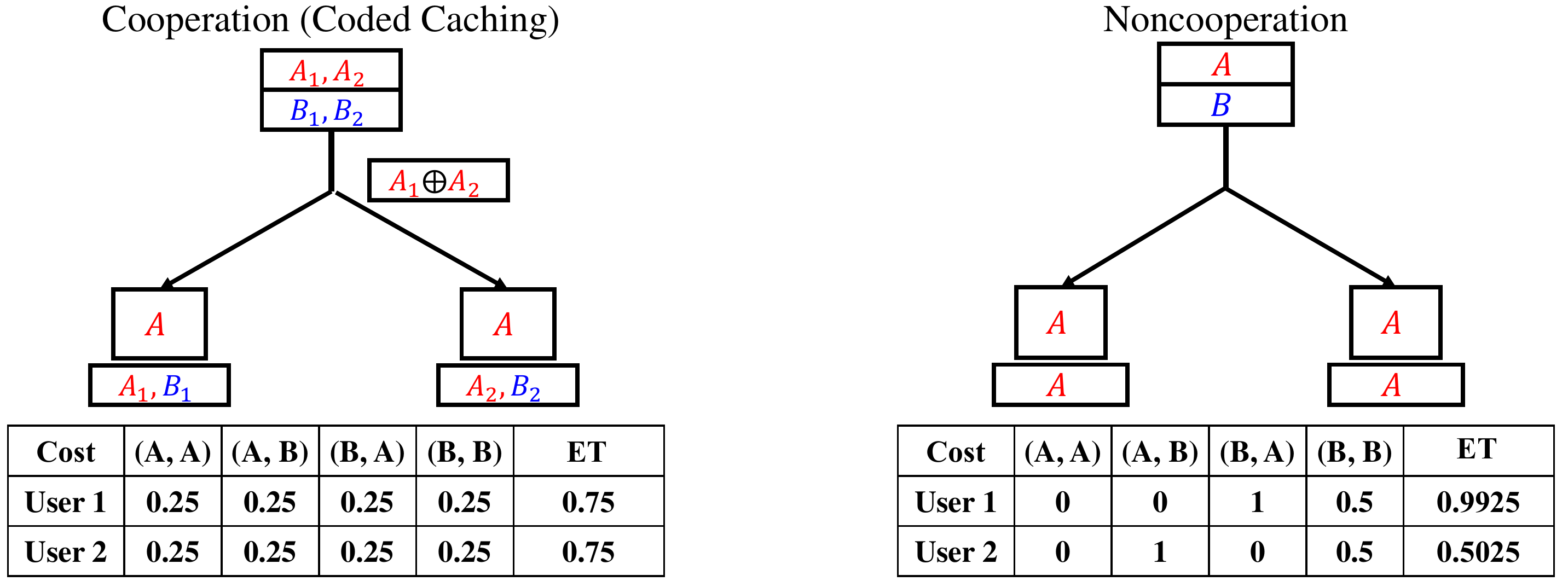}
      \caption{Effective throughputs under different caching schemes. The pair (A,A) in the tables stands for the realization when both the two users request content $A$ and the other pairs are similar. ET is short for effective throughput.  }    
      \label{preferencecost}  
  \end{figure}

  \section{Property of the Achievable Domain}
  
  In this section, we prove the convexity of the achievable domain of effective throughputs   and then focus on a special type of placement and delivery policies, termed UPAF  policies. It will be shown that  $(\bm{B},N,\mathcal{P})$-Caching with two users has an achievable domain as a polygon under UPAF policies. 
  
\subsection{Convexity of  the Achievable Domain}
The following theorem presents the convexity of the achievable domain of effective throughputs  .

\begin{theorem}\label{convexity}
   For any $(\bm{B},N,\mathcal{P})$-Caching, the achievable domain $ \mathcal{R}$ is a convex set and for any point $(R_k)_{k\in [K]}\in \mathcal{R}$,
   \begin{equation}\label{left_bottom_point}
       \{(x_k)_{k\in[K]}:0\le x_k\le R_k\}\subseteq \mathcal{R}.
   \end{equation}
\end{theorem} 
\begin{figure} 
    \centering
    \includegraphics[width=5.31cm]{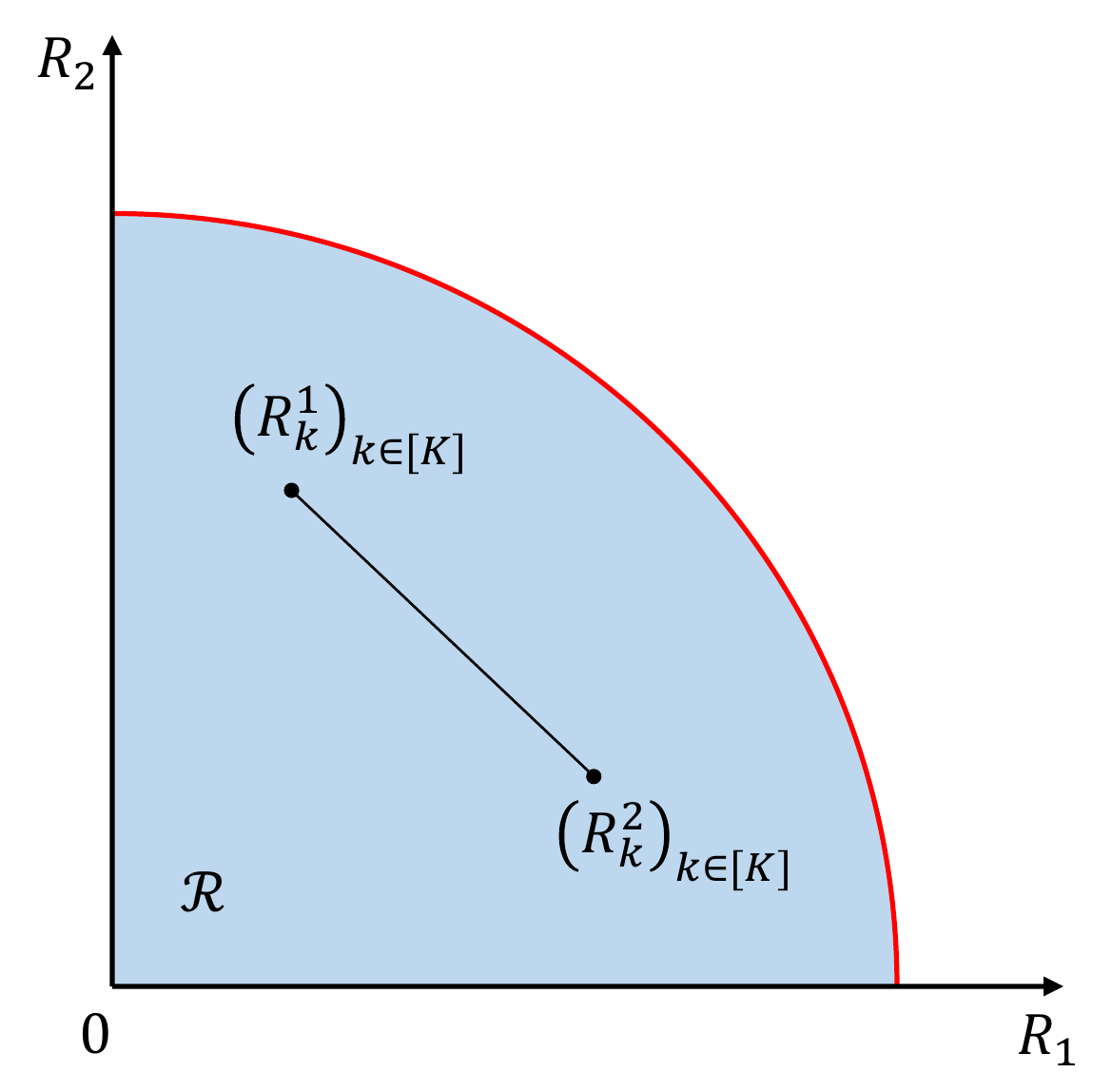}
    \caption{Convex set in the positive orthant. This set can be completely characterized by its boundary in the positive orthant.}    
    \label{convex}  
\end{figure}

\begin{proof}
Let $(R_k^1)_{k\in [K]}$ and $(R_k^2)_{k\in [K]}$ be two points in $\mathcal{R}$, as illustrated in Fig. \ref{convex}. To prove the convexity, we only need to show that $(\alpha R_k^1+ (1-\alpha)R_k^2)_{k\in [K]}$ also belongs to $\mathcal{R}$ for any $\alpha \in (0,1)$. Given $\varepsilon>0$ and $F$ large enough, suppose $\pi_1=\pi_1(\varepsilon,\alpha F)$ and $\pi_2=\pi_2(\varepsilon,(1-\alpha) F)$ achieve $(R_k^1)_{k\in [K]}$ and $(R_k^2)_{k\in [K]}$ with error probability lower than $\varepsilon$. We divide each content item into two portions, containing $\alpha F$ bits and $(1-\alpha )F$ bits respectively. By further dividing each buffer into two portions of size $\alpha b_kF$ bits and $(1-\alpha )b_kF$ bits, the original $(\bm{B},N,\mathcal{P})$-Caching can be viewed as two  $(\bm{B},N,\mathcal{P})$-Caching with different content sizes. The policies $\pi_1$ and $\pi_2$ can be respectively applied in the two  $(\bm{B},N,\mathcal{P})$-Caching. The error probability is bounded by $1-(1-\varepsilon)^2=2\varepsilon-\varepsilon^2$. The effective throughput of user $k$ is given by
\begin{equation*}
\begin{split}
R_k= &\sum_{n\in [N]}p_{k,n}-\mathbb{E}_{\bm{D}}\left(\sum_{\mathcal{U}:k\in \mathcal{U}}\frac{|M^1_{\mathcal{U}}|+|M^2_{\mathcal{U}}|}{|\mathcal{U}| F}\right)\\
=& \alpha \!\left(\sum_{n\in [N]}p_{k,n}-\mathbb{E}_{\bm{D}}\left(\sum_{\mathcal{U}:k\in \mathcal{U}}\frac{|M^1_{\mathcal{U}}|}{|\mathcal{U}|\alpha F}\right)\right) \!+\! (1-\alpha)\! \left(\sum_{n\in [N]}p_{k,n}-\mathbb{E}_{\bm{D}}\left(\sum_{\mathcal{U}:k\in \mathcal{U}}\frac{|M^2_{\mathcal{U}}|}{|\mathcal{U}|(1-\alpha) F}\right)\right)\\
 =& \alpha R_k^1 + (1-\alpha )R_k^2,
\end{split}
\end{equation*}
where $M^1_{\mathcal{U}}$ and $M^2_{\mathcal{U}}$ are messages generated by $\pi_1$ and $\pi_2$ respectively. For any $\alpha \in (0,1)$, $(\alpha R_k^1+ (1-\alpha)R_k^2)_{k\in [K]}$ can be achieved by 
combining $\pi_1$ and $\pi_2$ for sufficiently large $F$. Hence, $\mathcal{R}$ is a convex set.

Eq. (\ref{left_bottom_point}) implies that if a point $(R_k)_{k\in [K]}$ is in $ \mathcal{ R}$, any point $(x_k)_{k\in [K]}$  such that  $0\le x_k\le R_k$ also belongs to  $\mathcal{ R}$. We prove that by constructing a policy achieving $(x_k)_{k\in [K]}$. Suppose $\pi(\varepsilon,F)$ achieves $(R_k)_{k\in [K]}\in \mathcal{ R}$ with error probability lower than $\varepsilon$. In the delivery phase, $\pi(\varepsilon,F)$ transmits $M_{\{k\}}$ to user user $k$ exclusively. Let us consider a new policy in which the BS additionally transmits a random message of $(R_k-x_k)F$ bits to user $k$. For this new policy, the error probability remains unchanged while the effective throughputs reduce to $(x_k)_{k\in [K]}$. Thus,  $(x_k)_{k\in [K]}$ is achievable. 
\end{proof}

 According to Theorem \ref{convexity}, we only need to pay attention to the boundary in the positive orthant in order to characterize $\mathcal{ R}$, as illustrated in Fig. \ref{convex}. In addition, we can have more insights on the achievable domain of $(\bm{B},N,\mathcal{P})$-Caching. The values of different users' effective throughputs  are  interchangeable. When a policy brings one user a high effective throughput, the other users may only obtain low effective throughputs.

It is intractable to investigate all the feasible policies and characterize the whole achievable domain of effective throughputs. In the paper, we focus on absolutely fair (AF) policies.
\begin{definition}
    A policy $\pi$ is \emph{absolutely fair} if each user in $\mathcal{U}$ can  obtain  the same amount of useful information from the   message $M_\mathcal{U}$, i,e,
    \begin{equation}
        I(M_\mathcal{U}; \{W_n:n\in \hat{D}_k\}|C_k)= I(M_\mathcal{U}; \{W_n:n\in \hat{D}_j\}|C_j)
    \end{equation}
    for $k,j\in \mathcal{U}$.\footnote{ $I(\cdot;\cdot|\cdot)$ represents the  conditional mutual information.}
\end{definition}
 By enforcing other users receive the bits needed only by a certain user, the transmission cost of this user can be simply reduced. In an AF policy, such enforcement is forbidden. In other words, all the users are fair and no one incurs  transmission costs for other users.  If an AF policy has an uncoded placement process, the policy is referred to as a  UPAF policy. 
 In the rest of the paper, our attention will be paid to UPAF policies.  The significance of studying UPAF policies is twofold. Theoretically, UPAF policies provide an inner bound for the whole achievable domain. In addition, UPAF policies are practical, since all the users are fairly treated and the placement process is simpler compared with a coded one. 

\subsection{ The Achievable Domain for $(\bm{B},N,\mathcal{P})$-Caching with Two Users }\label{twousersection}

In this subsection, we investigate the achievable domain of $(\bm{B},N,\mathcal{P})$-Caching with two users.
A policy $\pi$ maps the content items and requests into the buffer states $(C_k)_{k\in [K]}$ and the transmitted messages $(M_{\mathcal{U}})_{\mathcal{U}\subseteq[K]}$. Thus we can use $(C_k)_{k\in [K]}$ and $(M_{\mathcal{U}})_{\mathcal{U}\subseteq[K]}$ to represent the policy $\pi$. To characterize $\mathcal{ R}$ for the  two-user case, we only need to study all  feasible $C_1,C_2,M_{\{1\}}, M_{\{2\}},$ and $M_{\{1,2\}}$. 

\begin{figure} 
    \centering
    \includegraphics[width=8.1cm]{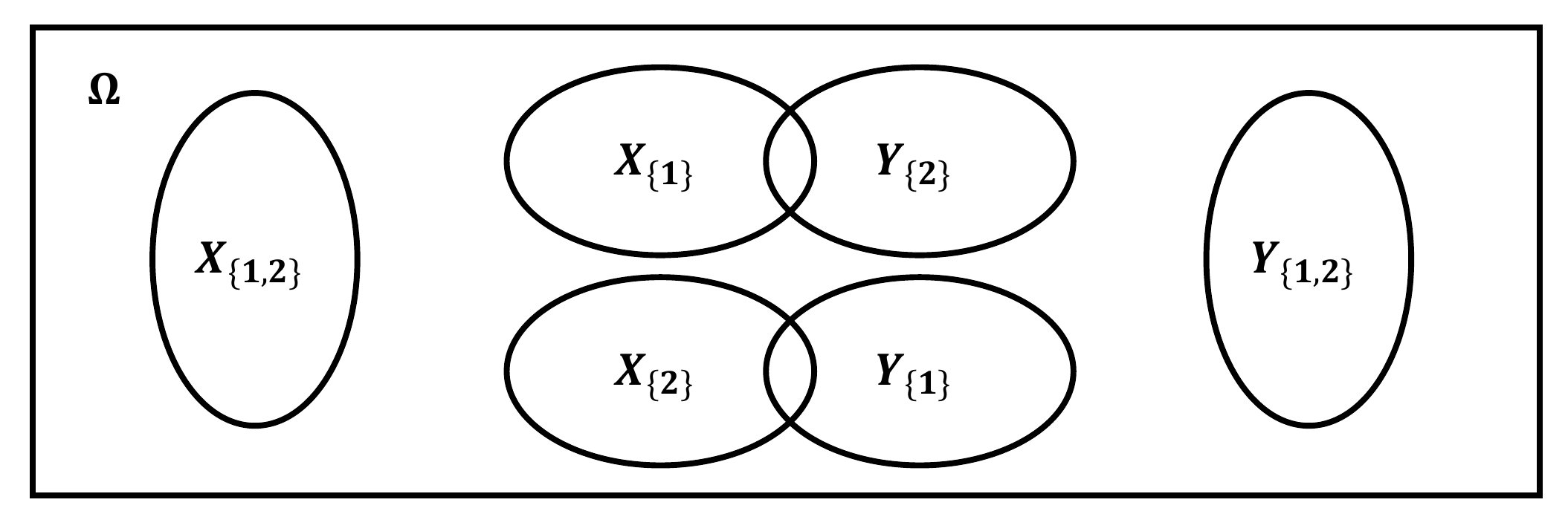}
    \caption{Venn Diagram of $X_{\mathcal{U}}$ and $Y_{\mathcal{U}}$.}    
    \label{venn}  
\end{figure}

For a UPAF policy, we have $C_1,C_2\subseteq \Omega$. Let us define 
\begin{eqnarray}
    X_{\emptyset} &=& \Omega\setminus(C_1\cup C_2),\\
    X_{\{1\}} &=& C_1\setminus C_2,\\
    X_{\{2\}} &=& C_2\setminus C_1,\\
    X_{\{1,2\}} &=& C_1 \cap C_2.
\end{eqnarray}
Then, $X_{\mathcal{U}}$ stands for the bits exclusively cached in the buffer of  user $k$ for $k\in \mathcal{U}.$ It is seen that $C_1 = X_{\{1\}}\cup X_{\{1,2\}}$ and $C_2 = X_{\{2\}}\cup X_{\{1,2\}}$. In the delivery phase, the demand matrix $\bm{\hat{D}}$ is known. Then, user $1$ and user $2$ wish to recover the bits in $Q_1=\{W_n: n\in \hat{D}_1\}$ and $Q_2=\{W_n:  n\in \hat{D}_2\}$, respectively. We define 
\begin{eqnarray}\label{req_set1}
Y_{\{1\}} &=& \left(Q_1\setminus C_1\right)\setminus \left(Q_2\setminus C_2\right),\\ \label{req_set2} 
Y_{\{2\}} &=& \left(Q_2\setminus C_2\right)\setminus \left(Q_1\setminus C_1\right),\\ \label{req_set3}
Y_{\{1,2\}} &=& \left(Q_1\setminus C_1\right) \cap \left(Q_2\setminus C_2\right).
\end{eqnarray}
Then $Y_{\mathcal{U}}$ is the set of bits requested only by users in $\mathcal{U}$ but not cached in the buffers of these users. As a result, users in $\mathcal{U}$ want to recover $Y_{\mathcal{U}}$ from the transmitted messages. Fig. \ref{venn} illustrates the relations between $X_{\mathcal{U}}$ and $Y_{\mathcal{U}}$.

To satisfy  user requests in the delivery phase, we can simply set $M_{\{1\}}=Y_{\{1\}}, M_{\{2\}}=Y_{\{2\}}, $ and $M_{\{1,2\}}=Y_{\{1,2\}}$. Notice that user $1$ contains a part of bits that user $2$ requests and vice versa. Index coding can be applied to create more multicasting opportunities and improve the effective throughputs of both user $1$ and user $2$ \cite{indexcoding}. In this case, the transmission costs of   user $1$ and user $2$ are given by 
\begin{eqnarray}
\begin{split}\label{through1}
 \bar{R}_1(C_1,C_2,\bm{\hat{D}})=&\frac{|Y_{\{1\}}| }{F}+ \frac{1}{2F}|Y_{\{1,2\}}| - \frac{1}{2F}\times\min\left\lbrace |X_{\{1\}}\cap Y_{\{2\}}|,|X_{\{2\}}\cap Y_{\{1\}} | \right\rbrace,
\end{split}   \\
\begin{split}\label{through2}
\bar{R}_2(C_1,C_2,\bm{\hat{D}})=&\frac{|Y_{\{2\}}| }{F}+ \frac{1}{2F}|Y_{\{1,2\}}| - \frac{1}{2F}\times\min\left\lbrace |X_{\{1\}}\cap Y_{\{2\}}|,|X_{\{2\}}\cap Y_{\{1\}} | \right\rbrace.
\end{split}  
\end{eqnarray}
Note that Eqs. (\ref{through1}) and (\ref{through2}) give the minimum transmission costs for given   $C_1, $ $C_2,$ and   $\bm{\hat{D}}$, because all the possible multicasting opportunities have been exploited.

The transmission costs, i.e., Eqs. (\ref{through1}) and (\ref{through2}),  can also be represented as functions of the two users' buffer states. To this end, we denote $X_{\mathcal{U},n}=X_{\mathcal{U}}\cap W_n$.  According to De Morgan's laws,   Eqs. (\ref{req_set1})-(\ref{req_set3}) can be rewritten as 
\begin{eqnarray}\label{re_req_set1}
Y_{\{1\}} &=&\left( \bigcup\limits_{n\in \hat{D}_1} X_{\{2\},n}\right)\bigcup \left( \bigcup\limits_{n\in \hat{D}_1\setminus \hat{D}_2} X_{\emptyset,n}\right),\\ \label{re_req_set2} 
Y_{\{2\}} &=& \left( \bigcup\limits_{n\in \hat{D}_2} X_{\{1\},n}\right)\bigcup \left( \bigcup\limits_{n\in \hat{D}_2\setminus \hat{D}_1} X_{\emptyset,n}\right), \\ \label{re_req_set3}
Y_{\{1,2\}} &=&  \bigcup\limits_{n\in \hat{D}_1\cap \hat{D}_2} X_{\emptyset,n}.
\end{eqnarray}
 Defining $x_{\mathcal{U},n}=\frac{|X_{\mathcal{ U},n}|}{F}$ and substituting Eqs. (\ref{re_req_set1})-(\ref{re_req_set3}) into Eqs. (\ref{through1})-(\ref{through2}) yield
\begin{eqnarray}
\begin{split}\label{re_through1}
\bar{R}_1(C_1,C_2,\bm{\hat{D}})=& \sum_{n\in \hat{D}_1} x_{\{2\},n}  +\sum_{n\in \hat{D}_1\setminus  \hat{D}_2} x_{\emptyset,n} + \frac{1}{2}\sum_{n\in \hat{D}_1\cap   \hat{D}_2} x_{\emptyset,n} - \frac{1}{2}\times \\&\min\left\lbrace \sum_{n\in \hat{D}_2} x_{\{1\},n},\sum_{n\in \hat{D}_1} x_{\{2\},n} \right\rbrace,
\end{split}   
\end{eqnarray}
\begin{eqnarray}
\begin{split}\label{re_through2}
\bar{R}_2(C_1,C_2,\bm{\hat{D}})=& \sum_{n\in \hat{D}_2} x_{\{1\},n} +\sum_{n\in \hat{D}_2\setminus  \hat{D}_1} x_{\emptyset,n}  + \frac{1}{2}\sum_{n\in \hat{D}_1\cap   \hat{D}_2} x_{\emptyset,n} - \frac{1}{2}\times\\&\min\left\lbrace \sum_{n\in \hat{D}_2}x_{\{1\},n} ,\sum_{n\in \hat{D}_1} x_{\{2\},n} \right\rbrace.
\end{split}  
\end{eqnarray}
The variable $x_{\mathcal{U},n} $ stands for how much $X_{\mathcal{U},n}$ accounts for $W_n$. 

For fixed uncoded placement $C_1$ and $C_2$, the maximum effective throughputs achieved by UPAF policies are given by 
\begin{eqnarray}
\label{eff1}
    R_1(C_1,C_2) =& \sum_{n\in [N]}p_{1,n}-\mathbb{E}_{\bm{D}}\left(\bar{R}_1(C_1,C_2,\bm{D})\right),\\    
    \label{eff2}
    R_2(C_1,C_2) =& \sum_{n\in [N]}p_{2,n}-\mathbb{E}_{\bm{D}}\left(\bar{R}_2(C_1,C_2,\bm{D})\right). 
\end{eqnarray}
From the proof of Theorem \ref{convexity}, one can see that Theorem  \ref{convexity} also holds for UPAF policies. Thus, the boundary of the achievable domain  under UPAF policies can be given by solving the following optimization problem:
\begin{equation}\label{op_1}
    \begin{split}
    \max_{x_{\mathcal{U},n}}\quad  &\alpha R_1(C_1,C_2) + (1-\alpha) R_2(C_1,C_2) \\
    \text{s.t.}\quad & \sum_{\mathcal{U}\subseteq \{1,2\}}x_{\mathcal{U},n} \le 1,  n\in [N],\\
    & \sum_{\mathcal{U}:k\in \mathcal{U}}\sum_{n\in [N]}x_{\mathcal{U},n} \le b_k,  k\in \{1,2\},\\
    & x_{\mathcal{U},n} \ge 0,\mathcal{U}\subseteq\{1,2\},n\in [N].
    \end{split}
\end{equation}
The first constraint is due to the fact that the sets $x_{\mathcal{U},n}$ are disjoint for $\mathcal{U}\subseteq \{1,2\}$. The second constraint forbids  buffer overflows in the placement phase. By tuning $\alpha$ from 0 to 1, we obtain the boundary of the achievable domain of effective throughputs   under UPAF policies.

 Problem (\ref{op_1}) is almost a linear programming (LP) problem except the minimizing operations. To simplify problem (\ref{op_1}), we introduce auxiliary variables $z_{D_1,D_2} = \min\left\lbrace \sum_{n\in D_2}x_{\{1\},n} ,\right.$ $\left.\sum_{n\in D_1} x_{\{2\},n} \right\rbrace$ and vectors
 \begin{eqnarray}
 \bm{x}_{1}&=& [x_{\{1\},1} + x_{\{1,2\},1},...,x_{\{1\},N}+ x_{\{1,2\},N}], \\
 \bm{x}_{2}&= & [x_{\{2\},1} + x_{\{1,2\},1},...,x_{\{2\},N}+ x_{\{1,2\},N}],\\
 \bm{x}_{3}&=& [x_{\{1,2\},1},...,x_{\{1,2\},N}], \\  
 \bm{z}& =& [z_{\emptyset,\emptyset},...,z_{[N],[N]},1].
 \end{eqnarray}
 Then $\bm{x}_1$ and  $\bm{x}_2$ represent the proportion of each content item  cached in user 1 and user 2, respectively. The vector $ \bm{x}_3$ represents the proportion of each content item cached in both the two users' buffers. The vector $\bm{z}$ is a function of $\bm{x}_1,\bm{x}_2,$ and $ \bm{x}_3$. The last element of $\bm{z}$ is a fixed constant $1$, which is used to express the constant term in Eqs. (\ref{eff1})-(\ref{eff2}) later.  The constraints on $z_{D_1,D_2}$ are linear:
   \begin{eqnarray}
   \label{zcon1}
 z_{D_1,D_2}\le \sum_{n\in D_2}x_{\{1\},n},\\
 \label{zcon2}
  z_{D_1,D_2}\le \sum_{n\in D_1} x_{\{2\},n}.
 \end{eqnarray} 
 Then, problem (\ref{op_1}) can be transformed into an equivalent LP problem\footnote{Let $\bm{v_1}=[v_{11},...,v_{1M}]$ and $ \bm{v_2}=[v_{21},...,v_{2M}]$ be two $M$-dimensional vectors. By $\bm{v}_1\le \bm{v_2}$, we mean $v_{1m}\le v_{2m}$ for $m\in [M]$. }
  \begin{equation}\label{op_2}
 \begin{split}
 \max_{\bm{x}_1,\bm{x}_2,\bm{x}_3,\bm{z}}\quad  &\alpha \left( \sum_{i=1}^{3} \bm{a}_{1,i}^T \bm{x}_{i} + \bm{b}_1^T \bm{z} \right) + (1-\alpha )\left(  \sum_{i=1}^{3} \bm{a}_{2,i}^T \bm{x}_{i} + \bm{b}_2^T \bm{z}  \right)  \\
 \text{s.t.}\quad & \sum_{i=1}^3\bm{A}_i\bm{x}_i+\bm
{Bz}\le \bm{h},
 \end{split}
 \end{equation}
 where $\left( \sum_{i=1}^{3} \bm{a}_{k,i}^T \bm{x}_{i} + \bm{b}_k^T \bm{z} \right)$ gives the effective throughput of user $k$ and $\bm{A}_i, \bm{B}$, as well as $\bm{h}$ result from the constraints in problem (\ref{op_1}) and Eqs. (\ref{zcon1})-(\ref{zcon2}). It should be noted that the coefficients $\bm{a}_{1,i}, \bm{b}_1, \bm{a}_{2,i}, \bm{b}_2$ depend on and only on the user preferences. 
Based on problem (\ref{op_2}), we have the following theorem.

\begin{theorem}\label{theo2}
    For $(\bm{B},N,\mathcal{P})$-Caching with two users, the achievable domain under UPAF policies is a polygon. 
\end{theorem}

\begin{proof}
 To show the achievable domain under UPAF policies is a polygon, we only need to show  its boundary is piecewise linear.   The points in the boundary can be obtained by solving problem (\ref{op_2}).  Notice that the constraints of problem (\ref{op_2}) are linear and  are independent of $\alpha$. Therefore, the feasible domain of  problem (\ref{op_2}) remains unchanged with different values of $\alpha$. 
 
 According to the LP theory, the feasible domain of an LP problem is a  convex polytope and the optimal solution is a vertex of the convex polytope \cite{lptheory}.\footnote{If an LP problem has only one optimal solution, this solution must be a vertex of the convex polytope. If an LP problem has multiple optimal solutions, at least one of the optimal solutions is a vertex of the convex polytope.} Since the number of  vertices of a convex polytope is finite, problem (\ref{op_2}) at most achieves finitely many different optimal solutions when $\alpha$ goes from 0 to 1. As a result, the boundary can be characterized by finitely many points and therefore is piecewise linear.     
\end{proof}

The achievable domain gathers all the possible values of effective throughputs that can be achieved by possible placement and delivery policies, no matter it is centralized or decentralized. The higher the effective throughput a user obtains, the lower the transmission cost this user affords. The boundary of the achievable domain represents the Pareto-efficient  effective throughputs. If the users are selfish and each user only wants to maximize its own effective throughput, the users form  a  game relationship. In the next section, games are formulated to allocate caching gains for $(\bm{B},N,\mathcal{P})$-Caching with two users.

 \section{Games in $(\bm{B},N,\mathcal{P})$-Caching with Two Users}
 
In this section, a noncooperative game is formulated to investigate the equilibrium on effective throughputs for $(\bm{B},N,\mathcal{P})$-Caching with two users. Based on the noncooperative game, a cooperative game is studied to allocate the caching gains. Furthermore, a low-complexity algorithm is presented to provide a reasonable effective throughput allocation.
 
\subsection{Noncooperative  Game in $(\bm{B},N,\mathcal{P})$-Caching with Two Users}

In this subsection, we investigate a noncooperative game in $(\bm{B},N,\mathcal{P})$-Caching with two users. More specifically, we assume  that the two users  fill their buffers individually in the placement phase.  The BS satisfies the user requests in a manner that  the number of transmitted bits is minimized.  Each user wishes to maximize its own effective throughput from caching. It will be shown that the noncooperative game always has a mixed Nash equilibrium (NE). In addition, the noncooperative game has  pure strategy Nash equilibria (PSNEs) when the user preferences are similar.

  In the noncooperative game, the two users take the roles of players. The sets of bits cached in the user buffers in the placement phase, i.e., $C_1$ and $C_2$, act as strategies. Throughout this subsection, we  consider only UPAF policies. As a result, the strategy sets for this two users are given by $\mathcal{S}_1 =\{C_1 \subseteq \Omega : |C_1|\le b_1 F \} $ and $\mathcal{S}_2 =\{C_2 \subseteq \Omega : |C_2|\le b_2 F \} $. The payoffs for this two users are the effective throughputs resulting from caching and thus are presented in Eqs. (\ref{eff1}) and (\ref{eff2}), respectively. We denote the above noncooperative game as $G_0 = \{\mathcal{S}_1,\mathcal{S}_2;R_1,R_2\}$.

Nash's existence theorem guarantees that the noncooperative game in $(\bm{B},N,\mathcal{P})$-Caching with two users at least has a mixed NE \cite{nash}. 
\begin{theorem}\label{existence}
    $G_0$ has a mixed NE. 
\end{theorem}
\begin{proof}
    $G_0$ has finitely many players. In addition, the strategy sets are finite, i.e., $|\mathcal{S}_k|\le \binom{NF}{b_kF}$ for $k=1,2$. Thus $G_0$ is a finite game. According to Nash's existence theorem, $G_0$ has a mixed NE. 
\end{proof}
A pure strategy is a bit-by-bit decision over the content items. However, a mixed NE needs not to divide a bit into smaller parts. Instead, a mixed NE chooses pure strategies according to a certain distribution.  Having proved the existence of a mixed NE, we pay attention to PSNEs.  However,  it is computationally prohibitive to find a PSNE and corresponding payoffs for $G_0$, due to the fact that the strategy sets are of exponential  sizes. To overcome that, we construct an infinite game based on $G_0$.

Note that the strategies $C_1$ and $C_2$ can be completely characterized by $\bm{x}_1, $ $\bm{x}_2, $ and $\bm{x}_3$.  The payoff functions in $G_0$ can be rewritten as 
\begin{eqnarray}\label{throughput1}
    R_k(\bm{x}_1, \bm{x}_2, \bm{x}_3,\bm{z})&=&\sum_{i=1}^{3} \bm{a}_{k,i}^T \bm{x}_{i} + \bm{b}_k^T \bm{z} .
\end{eqnarray}
Let us consider a two-player infinite game $G_1 = \{\mathcal{E}_1, \mathcal{E}_2; f_1,f_2\}$. The strategy sets are feasible domains of $\bm{x}_1$ and $\bm{x}_2$, i.e., $\mathcal{E}_1 = \{\bm{x}\in R^N:\bm{0}\le \bm{x}\le \bm{1}, \sum_{n=1}^Nx_n \le b_1 \}$ and $\mathcal{E}_2 = \{\bm{x}\in R^N:\bm{0}\le \bm{x}\le \bm{1}, \sum_{n=1}^Nx_n \le b_2 \}$. The payoff function $f_k$ is defined as 
\begin{equation}
\begin{split}
    f_k (\bm{x}_1, \bm{x}_2) &= \max_{\bm{x}_3,\bm{z}} R_k(\bm{x}_1, \bm{x}_2, \bm{x}_3,\bm{z})\\
    &=   \bm{a}_{k,1}^T \bm{x}_1  +  \bm{a}_{k,2}^T \bm{x}_2 + g_k(\bm{x}_1, \bm{x}_2) ,
\end{split}
\end{equation}
where 
\begin{equation}\label{semipayoff}
    \begin{split}
    g_k(\bm{x}_1,\bm{x}_2) = &\max_{\bm{x}_3,\bm{z}} \bm{a}_{k,3}^T\bm{x}_3 + \bm{b}_k^T\bm{z}\\
    &\text{s.t. } \bm{A}_3\bm{x}_3 + \bm{Bz}\le \bm{h} - \bm{A}_1\bm{x}_1-\bm{A}_2\bm{x}_2,
    \end{split}
\end{equation}
where the constraint is equivalent to the one in problem (\ref{op_2}). 
     One can see that problem (\ref{semipayoff}) is also an LP problem.  The basic idea to formulate $G_1$ is as follows. Each user decides the number of bits  cached in its own buffer. Thus $\bm{x}_1$ and $\bm{x}_2$ act as strategies. When the values of $\bm{x}_1$ and $\bm{x}_2$ are selected, user 1 is granted the privilege to  maximize its own effective throughput by adjusting the values of $\bm{x}_3$ and $\bm{z}$ and then yields the payoff function $f_1(\bm{x}_1,\bm{x}_2)$. Similarly, $f_2(\bm{x}_1,\bm{x}_2)$ can be derived. 

Lemma \ref{lemma1}  presents the relationship between the existences of PSNEs of $G_0$ and $G_1$. 

\begin{lemma}\label{lemma1}
    If $G_1$ has a PSNE $(\bm{x}_1^*, \bm{x}_2^*)$ and there exist $\bm{x}_3^*$ and $\bm{z}^*$ satisfying $g_k (\bm{x}_1^*, \bm{x}_2^*) = \bm{a}_{k,3}^T\bm{x}_3^* + \bm{b}_k^T\bm{z}^*$, 
     $G_0$ has a PSNE.
\end{lemma}
\begin{proof}
 Suppose $(\bm{x}_1^*,\bm{x}_2^*)$ is a PSNE of $G_1.$ We have 
 \begin{eqnarray}
     f_1(\bm{x}_1^*,\bm{x}_2^*)  &\ge &  f_1(\bm{x}_1,\bm{x}_2^*)  , \text{ for }\bm{x}_1 \in \mathcal{E}_1,\\ 
     f_2(\bm{x}_1^*,\bm{x}_2^*)  &\ge &  f_2(\bm{x}_1^*,\bm{x}_2)  , \text{ for }\bm{x}_2 \in \mathcal{E}_2.
           \end{eqnarray}
           The condition $g_k (\bm{x}_1^*, \bm{x}_2^*) = \bm{a}_{k,3}^T\bm{x}_3^* + \bm{b}_k^T\bm{z}^*$ implies that problem (\ref{semipayoff}) has the same optimal solution $(\bm{x}_3^*,\bm{z}^*)$ for both $k=1$ and $k=2$.
 Thus, we have
  \begin{eqnarray}
  \label{nashcond1}
 R_1(\bm{x}_1^*,\bm{x}_2^*,\bm{x}_3^*,\bm{z}^*)  &\ge &  f_1(\bm{x}_1,\bm{x}_2^*) \ge R_1(\bm{x}_1,\bm{x}_2^*,\bm{x}_3,\bm{z})  , \\ 
 \label{nashcond2}
 R_2(\bm{x}_1^*,\bm{x}_2^*,\bm{x}_3^*,\bm{z}^*)  &\ge &  f_2(\bm{x}_1^*,\bm{x}_2) \ge R_2(\bm{x}_1^*,\bm{x}_2,\bm{x}_3,\bm{z}). 
 \end{eqnarray}
 Let us consider $(C_1^*,C_2^*)$ satisfying  $\frac{C_1^*\cap W_n}{F} = x_{1,n}^*, \frac{C_2^*\cap W_n}{F} = x_{2,n}^*,$ and $\frac{C_1^*\cap C_2^* \cap W_n}{F} = x_{3,n}^*$. Eqs. (\ref{nashcond1}) and (\ref{nashcond2}) can be written as 
  \begin{eqnarray}
R_1(C_1^*,C_2^*)  &\ge &  R_1(C_1,C_2^*)  , \\ 
R_2(C_1^*,C_2^*)  &\ge &   R_2(C_1^*,C_2). 
\end{eqnarray}
Thus $(C_1^*,C_2^*)$ is a PSNE of $G_0$.     
\end{proof}

Lemma \ref{lemma1} reveals that $G_0$ has a PSNE if $G_1$ has a PSNE and the two users adopt the same pair values of  $\bm{x}_3$ and $\bm{z}$. By applying Debreu's theorem \cite{debreu}, we can prove that $G_1$ has a PSNE. 

\begin{lemma}\label{lemma2}
    $G_1$ has a PSNE. 
\end{lemma}
\begin{proof}
    According to Debreu's theorem, we only need to show that the strategy sets are nonempty convex compact subsets of an Euclidean space and the payoff functions are continuous and quais-concave. 
    
    The sets $\mathcal{E}_1$ and $\mathcal{E}_2$ are bounded and closed in $R^N$ and thus are also compact. The payoff functions are maximums of a series of linear functions:
    \begin{equation}
  f_k (\bm{x}_1, \bm{x}_2) = \max_{\bm{x}_3,\bm{z}} \sum_{i=1}^{3} \bm{a}_{k,i}^T \bm{x}_{i} + \bm{b}_k^T \bm{z} .
    \end{equation}
    It is seen that $f_k (\bm{x}_1, \bm{x}_2) $ is continuous. To prove the quais-concavity of $f_k (\bm{x}_1, \bm{x}_2) $, 
     we only need to show $\sum_{i=1}^{3} \bm{a}_{k,i}^T \bm{x}_{i} + \bm{b}_k^T \bm{z}$ is quais-concave, according to the properties of quais-concavity. As a linear function, $\sum_{i=1}^{3} \bm{a}_{k,i}^T \bm{x}_{i} + \bm{b}_k^T \bm{z}$ is quais-concave. 
\end{proof}


Based on Lemmas \ref{lemma1} and \ref{lemma2}, we have the following theorem. 

\begin{theorem}\label{psne}
    Given any $\bm{B}$ and $N$, there exists a positive number $\varepsilon $ such that if  $||\bm{a}_{1,3}-\bm{a}_{2,3}|| + ||\bm{b}_1-\bm{b}_2||\le \varepsilon$,  $G_0$ has PSNEs and the PSNEs are not unique. 
\end{theorem}
\begin{proof}
According to  Lemmas \ref{lemma1} and \ref{lemma2}, we only need to show that  problem (\ref{semipayoff}) has the same optimal solution for $k=1$ and $k=2$ at the PSNE of $G_1$, when $||\bm{a}_{1,3}-\bm{a}_{2,3}|| + ||\bm{b}_1-\bm{b}_2||$ is small enough. Notice that the feasible domain of problem (\ref{semipayoff}) remains unchanged with different values of $k$. The objective function of problem (\ref{semipayoff}) describes a group of hyperplanes with the same normal vector $(\bm{a}_{k,3},\bm{b}_k)$. The optimal solution happens to be the   intersection point of the feasible domain and a certain hyperplane \cite{lptheory}. If the normal vectors for $k=1$ and $k=2$ are close enough, problem (\ref{semipayoff}) achieves its optimization value at the same vertex of its feasible domain. 
 Thus $G_0$ has a PSNE when the difference between $(\bm{a}_{1,3},\bm{b}_1)$ and $(\bm{a}_{2,3},\bm{b}_2)$ is small.

Having proved the existence of PSNEs, we now show its non-uniqueness. Let us consider the dual problem of problem (\ref{semipayoff}):
\begin{equation}\label{dual}
\begin{split}
\min_{\bm{\lambda}\ge \bm{0}}\,\,  & \left(\bm{h} - \bm{A}_1\bm{x}_1-\bm{A}_2\bm{x}_2 \right)^T\bm{\lambda} \\
\text{s.t. } &\bm{A}_3^T\bm{\lambda} = \bm{a}_{k,3},\\
& \bm{B}^T\bm{\lambda} = \bm{b}_{k}.
\end{split}
\end{equation}
Problem (\ref{dual}) is also an LP problem. Let $\bm{\lambda}_k$ denote the optimal solution of problem (\ref{dual}). Then the payoff functions can be written as $f_k=\bm{a}_{k,1}^T \bm{x}_1  +  \bm{a}_{k,2}^T \bm{x}_2 - \bm{\lambda}^T_k( \bm{A}_1\bm{x}_1+\bm{A}_2\bm{x}_2-\bm{h} )$. According to the LP theory,  $\bm{\lambda}_k$ is a  step function with respect to $\bm{x}_1$ and $\bm{x}_2$ \cite{lptheory}. In other words, we have 
\begin{eqnarray}\label{part1}
    \frac{\partial \bm{\lambda}_k}{\partial \bm{x}_1}=\bm{0},\\
    \label{part2}
       \frac{\partial \bm{\lambda}_k}{\partial \bm{x}_2}=\bm{0},    
\end{eqnarray}
almost everywhere. 

Suppose $(\bm{x}_1^*,\bm{x}_2^*)$ is a PSNE of $G_1$. We require and only require
\begin{eqnarray}
\label{necond1}
\left.\frac{\partial f_1}{\partial \bm{x}_1}\right | _{(\bm{x}_1^*,\bm{x}_2^*)}=\bm{0},\\
\label{necond2}
\left.\frac{\partial f_2}{\partial \bm{x}_2}\right | _{(\bm{x}_1^*,\bm{x}_2^*)}=\bm{0}.
\end{eqnarray}
Substituting Eqs. (\ref{part1})-(\ref{part2}) into Eqs. (\ref{necond1})-(\ref{necond2}) yields 
\begin{eqnarray}
\bm{A}_1^T\bm{\lambda}_1=\bm{a}_{1,1},\\
\bm{A}_2^T\bm{\lambda}_2=\bm{a}_{2,1},
\end{eqnarray}
which are  sufficient and necessary conditions for PSNEs.
Since $\bm{\lambda}_k$ is a step function of $\bm{x}_1$ and $\bm{x}_2$, $G_1$ have more than one PSNE, which further certifies that $G_0$ has more than one PSNE. 

\end{proof}

 Since the normal vectors of the hyperplanes are only related to the user preferences, Theorem \ref{psne} implies that $G_0$ has PSNEs when the two users' preferences are similar enough. The buffer size vector has no effect on the existence of PSNEs of $G_0$. It is worth noting that Theorem \ref{psne} does not claim that $G_0$ has no PSNEs when the two users' preferences differ much.  Instead, simulations demonstrate that $G_0$ has PSNEs at most cases. 
 
Note that the proof of the existence of NE is nonconstructive. Now we pay attention to how to calculate a NE. Unfortunately, it have been proved that calculating a NE is PPAD-complete, which is a subclass of NP \cite{complexity}. The exponential size of the strategy sets also implies the difficulty of finding a NE. Recall that the core condition that $(\bm{x}_1,\bm{x}_2)$ forms a PSNE is that the two users reach a consensus on the value of $\bm{x}_3$ ($\bm{z}$ is a function of $\bm{x}_1, \bm{x}_2,$ and $\bm{x}_3$). Based on best response dynamics, we present Algorithm \ref{algo1} to give a PSNE with high probability. The core idea behind Algorithm \ref{algo1} is to update the two users' caching schemes, i.e., $\bm{x}_1$ and $\bm{x}_2$, alternately. Once the caching schemes converge and the two users reach the  the same $\bm{x}_3$ in Steps 4 and 5, a PSNE is found and the program terminates. 
 The parameter $T$ is a positive integer that limits the maximum number of iterations. 
 
 Even though simulations reveal that Algorithm \ref{algo1} can return a PSNE at most cases, it is likely that no PSNE is found when the program terminates. That might be  because $G_0$ does not have a PSNE,  a bad initialization point is chosen, the number of iterations is not large enough, or other reasons. There is a balance between the probability that Algorithm \ref{algo1} finds a PSNE and the computational complexity. Higher $T$ will increase the probability and computational complexity simultaneously. 
 
\begin{algorithm}[t]
    \caption{Find a PSNE for $(\bm{B},N,\mathcal{P})$-Caching}
    \label{algo1}
    \renewcommand{\algorithmicrequire}{ \textbf{Input:}} 
    \renewcommand{\algorithmicensure}{ \textbf{Output:}} 
    \begin{algorithmic}[1]
        \REQUIRE   $\bm{B},N, \mathcal{P}, T, \varepsilon$
        \ENSURE PSNE $(\bm{x}_1^*,\bm{x}_2^*,\bm{x}_3^*)$
        \STATE Calculate $\bm{a}_{k,i}$ and $\bm{b}_k$ according to $\bm{B},N, \mathcal{P}$;      
        \STATE Randomly initialize $\bm{x}_1^{*}$ such that $\bm{x}_1^{*} \ge 0$ and $\sum_{n=1}^{N}x_{1,n}^{*}\le b_1$;
        \FOR{$t=1$  to $T$ } 
        \STATE Maximize $R_2(\bm{x}_1,\bm{x}_2,\bm{x}_3,\bm{z})$ for fixed $\bm{x}_1 = \bm{x}_1^{*}$ and obtain the  optimal solution $\bm{x}_2= \bm{x}_2^*$ and $\bm{x}_3 = \bm{x}_{3}^*$;
        \STATE Maximize $R_1(\bm{x}_1,\bm{x}_2,\bm{x}_3,\bm{z})$ for fixed $\bm{x}_2 = \bm{x}_2^{*}$ and obtain the  optimal solution $\bm{x}_1= \bm{x}_1^{**}$ and $\bm{x}_3 = \bm{x}_{3}^{**}$;
        \IF {$||\bm{x}_1^* - \bm{x}_1^{**}|| + ||\bm{x}_3^* - \bm{x}_3^{**}||\le \varepsilon$}
        \STATE \textbf{break};
        \ELSE 
        \STATE $\bm{x}_1^* = \bm{x}_1^{**}$;
        \ENDIF
        \ENDFOR        
        \STATE \textbf{return} $(\bm{x}_1^*,\bm{x}_2^*,\bm{x}_3^*)$.
    \end{algorithmic}
\end{algorithm} 
 
 {\itshape Remark 1:} In order to construct $G_0$, the two users and the BS should know the user preferences $\mathcal{P}$. Once $\mathcal{P}$ is known and a NE is found, the two users will prefetch data according to the NE without communicating with each other. 

\subsection{Cooperative Game in $(\bm{B},N,\mathcal{P})$-Caching with Two Users}

In this subsection, a cooperative game in $(\bm{B},N,\mathcal{P})$-Caching with two users is investigated. A numerical algorithm is presented to find a PSNE for the two users. Then reasonable allocation schemes are proposed, which help in designing a pricing policy. 

If the two users do not cooperate with each other and the BS serves each user individually, i.e., $M_{\mathcal{U}}=\emptyset$ for $|\mathcal{U}|\ge 2$, the optimal caching strategy should be caching the content items that a user requests with the highest probabilities. This scheme is referred to as pure caching. The effective throughput resulting from pure caching is given by 
\begin{equation}\label{nocoop}
    R_k^p = \sum_{n=1}^{\lfloor C_k \rfloor}p_{k,(n)} + (C_k - \lfloor C_k \rfloor)p_{k, (\lceil C_k \rceil)},
\end{equation}
where $p_{k,(n)}$ is the $n$-th largest number among $p_{k,1},...,p_{k,N}$. If the two users does not cooperate in the placement phase, the payoffs at a NE in the noncooperative game, denoted by $R_1^n$ and $R_2^n$,  indicate their effective throughputs in the delivery phase.  If the two users cooperate with each other from the placement phase, additional profits (effective throughputs) can be obtained from their alignment.  In this case,  the two users do not exchange data  but their placement policy and delivery policy are designed jointly. Let $R^c$ denote the total effective throughput resulting from cooperation, which indicates the reduction in the number of bits transmitted by the BS due to user cooperation. The core problem in the cooperative game  is to design schemes to allocate the revenue $R^c$ such that the users are willing to  participate in cooperation.  

Let $\bm{y}=(y_1,y_2)$ be an allocation scheme. All the  allocation schemes that hold individual rationality and collective rationality form a core for the cooperative game
\begin{equation}
C(\bm{B},N, \mathcal{P}) = \{\bm{y}:y_1\ge R_1^n,y_2\ge R_2^n, y_1 + y_2 = R^c \}.
\end{equation}
Generally, the core for a cooperative game is likely to be empty \cite{core}. However, the following theorem reveals that  $C(\bm{B},N, \mathcal{P})$ is nonempty. 
\begin{theorem}\label{theo5}
    $C(\bm{B},N, \mathcal{P})$ is nonempty for any $\bm{B}$, $N$, and $\mathcal{P}$.
\end{theorem}
\begin{proof}
    To prove this theorem, we only need to show $R_1^n + R_2^n\le R^c$.
    According to Theorem \ref{convexity}, all the achievable values of effective throughputs form a convex set $\mathcal{R}$. Thus we have $(R_1^n,R_2^n)\in \mathcal{R}$ no matter the payoffs are achieved by a mixed NE or a PSNE. The effective throughput from cooperation is given by $R^c = \max_{(y_1,y_2)\in \mathcal{ R}} (y_1+ y_2),$  which ensures that $R_1^n + R_2^n\le R^c$.
\end{proof}

All the allocation schemes in $C(\bm{B},N, \mathcal{P})$ are acceptable for rational users. In cooperative games, the satisfaction degree of a player on an allocation scheme is usually measured by the difference between the payoff resulting from noncooperative game  and the allocation, i.e., $R_k^n - y_k$. The lower the difference, the higher the degree of satisfaction. By maximizing the minimum satisfaction degrees among the two users, we reach an unique allocation scheme 
\begin{equation}\label{shapley}
R_k^c =R_k^n + \frac{R^c - R_1^n -R_2^n}{2},
\end{equation}
which is usually referred to as the  nucleolus in cooperative games and happens to be the Shapley values \cite{shapley}. It is seen that $R_1^c +R_2^c =R^c$ and $R_k^c \ge R_k^n$ according to Theorem \ref{theo5}. 

\begin{algorithm}[t]
    \caption{Allocation Scheme for $(\bm{B},N,\mathcal{P})$-Caching}
    \label{algo2}
    \renewcommand{\algorithmicrequire}{ \textbf{Input:}} 
    \renewcommand{\algorithmicensure}{ \textbf{Output:}} 
    \begin{algorithmic}[1]
        \REQUIRE   $\bm{B},N, \mathcal{P}, T, \varepsilon$
        \ENSURE Allocation  $(R_1^c,R_2^c)$
        \STATE Calculate $R^c$ by setting $\alpha = 0.5$ in Eq. ( \ref{op_1} ) and doubling the  optimization value;
        \STATE Call Algorithm \ref{algo1} and obtain $(\bm{x}_1^*,\bm{x}_2^*,\bm{x}_3^*)$;      
        \IF{$(\bm{x}_1^*,\bm{x}_2^*,\bm{x}_3^*)$ is a PSNE}
        \STATE Calculate $R_1^n$ and $R_2^n$ according to Eq. (\ref{throughput1});
        \STATE Calculate $R_1^c$ and $R_2^c$ according to Eq. (\ref{shapley});
        \ELSE
        \STATE Calculate $R_1^p$ and $R_2^p$ according to Eq. (\ref{nocoop});
        \STATE Calculate $R_1^c$ and $R_2^c$ according to Eq. (\ref{shapley2});
        \ENDIF
        \STATE \textbf{return} $(R_1^c,R_2^c)$.
    \end{algorithmic}
\end{algorithm}

Considering that $R^c$ represents the overall effective throughput, we only need to set $\alpha = 0.5$ in problem (\ref{op_2}) and multiply the  optimization value by two in order to obtain $R^c$.  The payoffs $R_1^c $ and $R_2^c $ also play a key role in the allocation scheme presented in Eq. (\ref{shapley}) and can be obtained by Algorithm \ref{algo1} with high probability. For the case that no PSNE is found in Algorithm \ref{algo1}, we suggest to allocate the cooperation gain according to $R_k^p$. More specifically, the effective throughput is given by 
\begin{equation}\label{shapley2}
R_k^c =R_k^p + \frac{R^c - R_1^p -R_2^p}{2}.
\end{equation}
Again, we have $R_1^c +R_2^c =R^c$ and $R_k^c \ge R_k^p$. 

 Algorithm \ref{algo2} summarizes the process of constructing an allocation scheme for the cooperative game in $(\bm{B},N,\mathcal{P})$-Caching with two users. Algorithm  \ref{algo2} can give an allocation scheme within low complexity but might not be fair enough when no NE is found. The possible unfairness results from the inherent hardness of computing NEs for noncooperative games.

\section{$(\bm{B},N,\mathcal{P})$-Caching in the General Multiuser Case}
In this section, we investigate $(\bm{B},N,\mathcal{P})$-Caching in the general multiuser case. A noncooperative game for $(\bm{B},N,\mathcal{P})$-Caching with multiple users is formulated.  However, it is intractable to find the NEs of this noncooperative game due to its large scale and payoffs. As a result, an algorithm is presented to give a placement and delivery policy directly. 

We consider uncoded placement processes in this section. Then the buffer states satisfy $C_k\subseteq \Omega$ for $k\in [K]$. Again, the $K$ users act as players.  User $k$'s strategy set is given by $\mathcal{S}_k^\text{m}=\{C_k\in \Omega:|C_k|\le b_kF \}$. After a strategy profile $(C_1,...,C_K)\in \prod_{k\in [K]}\mathcal{S}_k^\text{m}$ is selected and the demand matrix $\bm{D}$ is revealed, the BS applies index coding to minimize the total effective throughput. The resulting effective throughput of user $k$ is denoted by $R_k^\text{m}(C_1,...,C_K,\bm{D})$. Thus the payoff of user $k$ is given by $R_k^\text{m}=\mathbb{E}_{\bm{D}}\left(R_k^\text{m}(C_1,...,C_K,\bm{D})\right)$. If the $K$ users prefetch data individually in the placement phase, they form a noncooperative game $G^\text{m} = \{ \mathcal{S}_1^\text{m},...,\mathcal{S}_K^\text{m};R_1^\text{m},...,R_K^\text{m}\}$. It is intractable to find the NEs of $G^\text{m}$ since the strategy sets are of exponential sizes and closed-form expressions of the payoffs are not derived. 

Instead of calculating a NE, we present a UPAF policy to organize user cooperation directly in the following. Let us denote
\begin{equation}
    X_{\mathcal{U}} =\bigcap_{k\in \mathcal{U}} C_k\setminus \bigcup_{k\notin \mathcal{U}}C_k.
\end{equation}
Then, $ X_{\mathcal{U}} $ denotes the set of bits cached exclusively in users in $\mathcal{U}$. The demand matrix $\hat{\bm{D}}$ is revealed in the delivery phase. The set of bits that user $k$ wants to recover is given by $Q_k = \{W_n,n\in \hat{D}_k\}$. We denote
\begin{equation}
\begin{split}
Y_{\mathcal{U},\mathcal{V}} =\left( \bigcap_{k\in \mathcal{U}} (Q_k\setminus C_k) \setminus \bigcup_{k\notin \mathcal{U}} (Q_k\setminus C_k) \right) \bigcap X_{\mathcal{V}}. 
\end{split}
\end{equation}
Then,  $Y_{\mathcal{U}, \mathcal{V}}$ stands for the set of bits requested only by users in $\mathcal{U}$ and cached only by users in $\mathcal{V}$. One can see that the sets $Y_{\mathcal{U},\mathcal{V}} $ are disjoint for different $\mathcal{U}$ or $\mathcal{V}$, and $Y_{\mathcal{U},\mathcal{V}} =  \emptyset $ if $\mathcal{U}\cap \mathcal{V} \neq  \emptyset.$

In the delivery phase, the requests for the content items can be decomposed into requests for the sets $Y_{\mathcal{U}, \mathcal{V}}$. The BS needs to transmits messages to help the users recover  $Y_{\mathcal{U},\mathcal{V}}$. A rough UPAF policy is just to set $M_{\mathcal{U}}=\bigcup_{\mathcal{V}\subseteq  [K]} Y_{\mathcal{U},\mathcal{V}}$. Notice that every user caches side information that others are interested in. Users can cooperate to obtain higher effective throughputs by index coding. We define 
\begin{equation}
    Z_{\mathcal{U},\mathcal{V}} = \bigcup_{\mathcal{S}:\mathcal{V}\subseteq  \mathcal{S}} Y_{\mathcal{U},\mathcal{S}}.
\end{equation}
Then, $Z_{\mathcal{U},\mathcal{V}} $ is the set of bits requested only by users in $\mathcal{U}$ and at least cached by users in $\mathcal{V}$. For two sets  $Z_{\mathcal{U}_1,\mathcal{V}_1}$ and $Z_{\mathcal{U}_2,\mathcal{V}_2}$,  it is seen that $Z_{\mathcal{U}_1,\mathcal{V}_1}\cap Z_{\mathcal{U}_2,\mathcal{V}_2}=\emptyset $  if  $\mathcal{U}_1\neq  \mathcal{U}_2$, and $Z_{\mathcal{U}_1,\mathcal{V}_1}\subseteq Z_{\mathcal{U}_2,\mathcal{V}_2}$ if $\mathcal{U}_1= \mathcal{U}_2$  and $\mathcal{V}_2\subseteq\mathcal{V}_1$. Let us consider $J$ set pairs $\{(\mathcal{U}_j,\mathcal{V}_j):j\in [J]\}$ satisfying
\begin{equation}\label{condition}
\begin{cases}
\mathcal{U}_i\subseteq \mathcal{V}_j,& i\neq j,\\
\mathcal{U}_i\cap \mathcal{V}_j = \emptyset,& i=j.
\end{cases}
\end{equation} 
It can be verified that Eq. (\ref{condition}) ensures that $Z_{\mathcal{U}_j,\mathcal{V}_j}$ are disjoint for $j\in [J]$.
 Let $Z_{\mathcal{U}_j,\mathcal{V}_j}^t$ denote the $t$-th bit in $Z_{\mathcal{U}_j,\mathcal{V}_j}$. By transmitting $\oplus_{j\in [J]}Z_{\mathcal{U}_j,\mathcal{V}_j}^t$ to users in $\bigcup_{j\in [J]}\mathcal{ U}_j$, users in $\mathcal{ U}_j$ can recover $Z_{\mathcal{U}_j,\mathcal{V}_j}^t$, because all the other bits have been cached in the buffers of users in $\mathcal{ U}_j$. The transmission cost of one bit is $\frac{1}{\sum_{j\in[J]}|\mathcal{U}_j|}$. To  maximize the effective throughputs, we need to carefully group $Z_{\mathcal{ U},\mathcal{V}}$. 

\begin{algorithm}[t]
    \caption{A UPAF Policy for the General Multiuser Case}
    \label{algo3}
    \renewcommand{\algorithmicrequire}{ \textbf{Placement Phase:}} 
    \renewcommand{\algorithmicensure}{ \textbf{Delivery Phase:}} 
    \begin{algorithmic}[1]
        \REQUIRE ~\\
        \STATE Cache the $  b_k $ most popular content items in the buffer of user $k$, i.e., $C_k=\{W_n:p_{n,k}\text{ is one of } b_k$ $ \text{largest numbers }\text{ for fixed }k \}$;
        \ENSURE ~\\
        \STATE Compute $Y_{\mathcal{U}, \mathcal{V}}$ and $Z_{\mathcal{U}, \mathcal{V}}$ according to  $\hat{\bm{D}}$;
        \STATE Initialization: $M_{\mathcal{ U}}=\emptyset$ for $\mathcal{ U}\subseteq [K]$;
        \FOR{$k=K,K-1,...,1$} 
        \FOR {$\mathcal{U}\subseteq [K]$ and  $|\mathcal{U}|=k$}
        
        \STATE Split $\mathcal{U}$ into $J$ disjoint subsets $(\mathcal{U}_j)_{j\in[J]}$ such that users in $\mathcal{U}_j$ ask for the same content items, i.e., $U_{k_1}$ and $U_{k_2}$ are classified into the same subset if $\hat{D}_{k_1}=\hat{D}_{k_2}$;
        \STATE $T=\min_{j\in [J]}|Z_{\mathcal{ U}_j,\mathcal{ U}\setminus \mathcal{ U}_j}|$;
        \STATE $M_{\mathcal{ U}}= M_{\mathcal{ U}}\cup \left\lbrace \oplus_{j\in [J]} Z_{\mathcal{ U}_j,\mathcal{ U}\setminus \mathcal{ U}_j}^t: t\in [T]\right\rbrace $;
        \FOR{$j\in [J]$}
        \FOR{$\mathcal{S}\subseteq \mathcal{ U}\setminus \mathcal{ U}_j$}
        \STATE $Z_{\mathcal{ U}_j,\mathcal{ S}} = Z_{\mathcal{ U}_j,\mathcal{ S}} \setminus \left\lbrace  Z_{\mathcal{ U}_j,\mathcal{ U}\setminus \mathcal{ U}_j}^t: t\in [T]\right\rbrace;$
        \ENDFOR
        \ENDFOR
        \ENDFOR
        \ENDFOR 
        \FOR {$\mathcal{U}\subseteq [K]$}
        \STATE $M_{\mathcal{U}} =M_{\mathcal{U}} \cup Z_{\mathcal{ U},\emptyset}; $
        \ENDFOR
        \STATE Multicast $M_{\mathcal{U}}$ to users in $\mathcal{ U}$.
        
    \end{algorithmic}
\end{algorithm}

 A grouping method for $Z_{\mathcal{ U},\mathcal{V}}$ is presented in Algorithm \ref{algo3}.  The $b_k$ most popular  content items are cached in the  buffer of user $k$ in the placement phase.\footnote{If $b_k$ is not an integer, user $k$ caches the $\lfloor b_k \rfloor$ most popular content items and $b_k-\lfloor b_k \rfloor$ of the $\lceil b_k \rceil$-th most popular content item in its buffer. } For each subset $\mathcal{ U}\subseteq [K]$, we create a group of set pairs $\{(\mathcal{ U}_j, \mathcal{U}\setminus \mathcal{ U}_j)\}$ according to the demand matrix $\hat{\bm{D}}$. Each coded bit is generated by taking a bit from $Z_{\mathcal{ U}_j,\mathcal{ U}\setminus \mathcal{ U}_j}$. In Steps 9 to 13, the sets $Z_{\mathcal{ U},\mathcal{V}}$ are updated. If some bits cannot be recovered from the coded bits, these bits are added to $M_{\mathcal{U}}$ in Steps 16 to 18. It is seen that Algorithm \ref{algo3} is absolutely fair, because each user in $\mathcal{ U}$ can recover $H(M_{\mathcal{U}})$ bits from $M_{\mathcal{U}}$.\footnote{ $H(\cdot)$ denotes the entropy.}  Since each user prefetches data individually in the placement phase, Algorithm \ref{algo3} is decentralized. 
 
 Let $R^3_k$ be the effective throughput of user $k$ from Algorithm 3. We have the following theorem. 
\begin{theorem}\label{theo6}
    $R^3_k\ge R_k^p$ for all  $k$.
\end{theorem}
\begin{proof}
   In Algorithm \ref{algo3}, users recover some bits from coded messages $ \oplus_{j\in [J]} Z_{\mathcal{ U}_j,\mathcal{ U}\setminus \mathcal{ U}_j}^t$. The transmission cost of bits recovered from these messages is $\frac{1}{|\mathcal{ U}|}$, lower than 1. Thus, $R_k^3\ge R_k^p$.
\end{proof}
Theorem \ref{theo6} indicates that  Algorithm \ref{algo3} provides a higher effective throughput for each user than pure caching does. Even though no NE is used in Algorithm \ref{algo3}, Algorithm \ref{algo3} is fair if it is executed enough times and the user preferences are randomly generated each time. This is because one user may gain extra advantage one time and may also suffer losses next time.

\section{Simulation Results}

In the section, we present numerical results to validate the theoretical analysis and illustrate the effective throughput gains from user cooperation. In the simulations, we assume that each user independently requests one content item in the delivery phase. Therefore, the probability measure $\mathcal{P}$ can be fully characterized by the user preference matrix $\bm{P}$.  Throughout this section, we always assume that the users are equipped with buffers of identical sizes, $b_k=B$ for $k\in [K]$. To validate the performance of the proposed algorithms, we compare them with pure caching. 


\begin{figure}
    \begin{minipage}[t]{0.5\linewidth}
        \centering
        \includegraphics[width=2.5in]{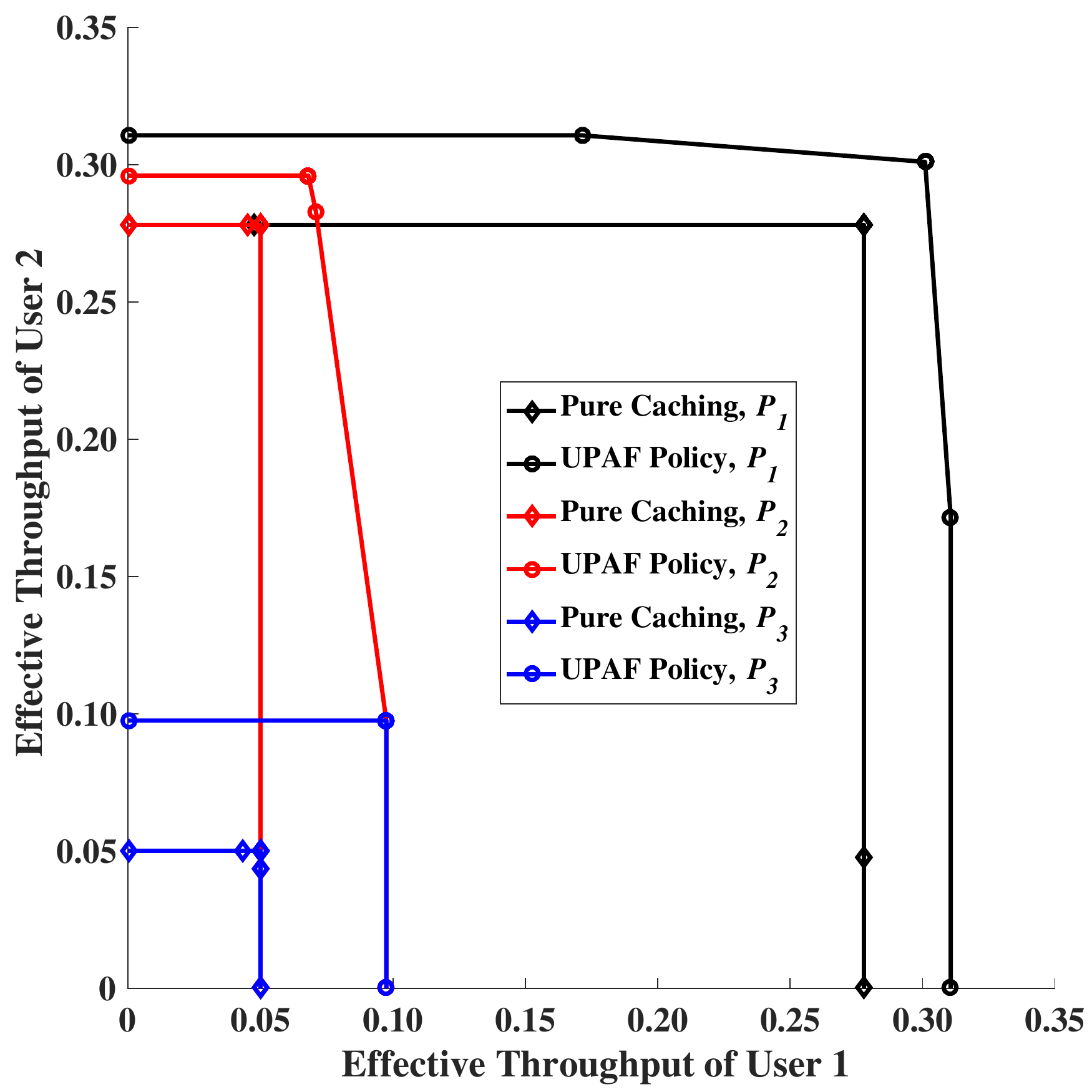}
 \caption{Achievable domain under UPAF policies for      $\hspace{0.1cm}$ $(\bm{B},N,\mathcal{P})$-Caching with two users and different preference$\hspace{0.4cm}$ matrices.}    
\label{different_p}  
    \end{minipage}%
    \begin{minipage}[t]{0.5\linewidth}
        \centering
        \includegraphics[width=2.5in]{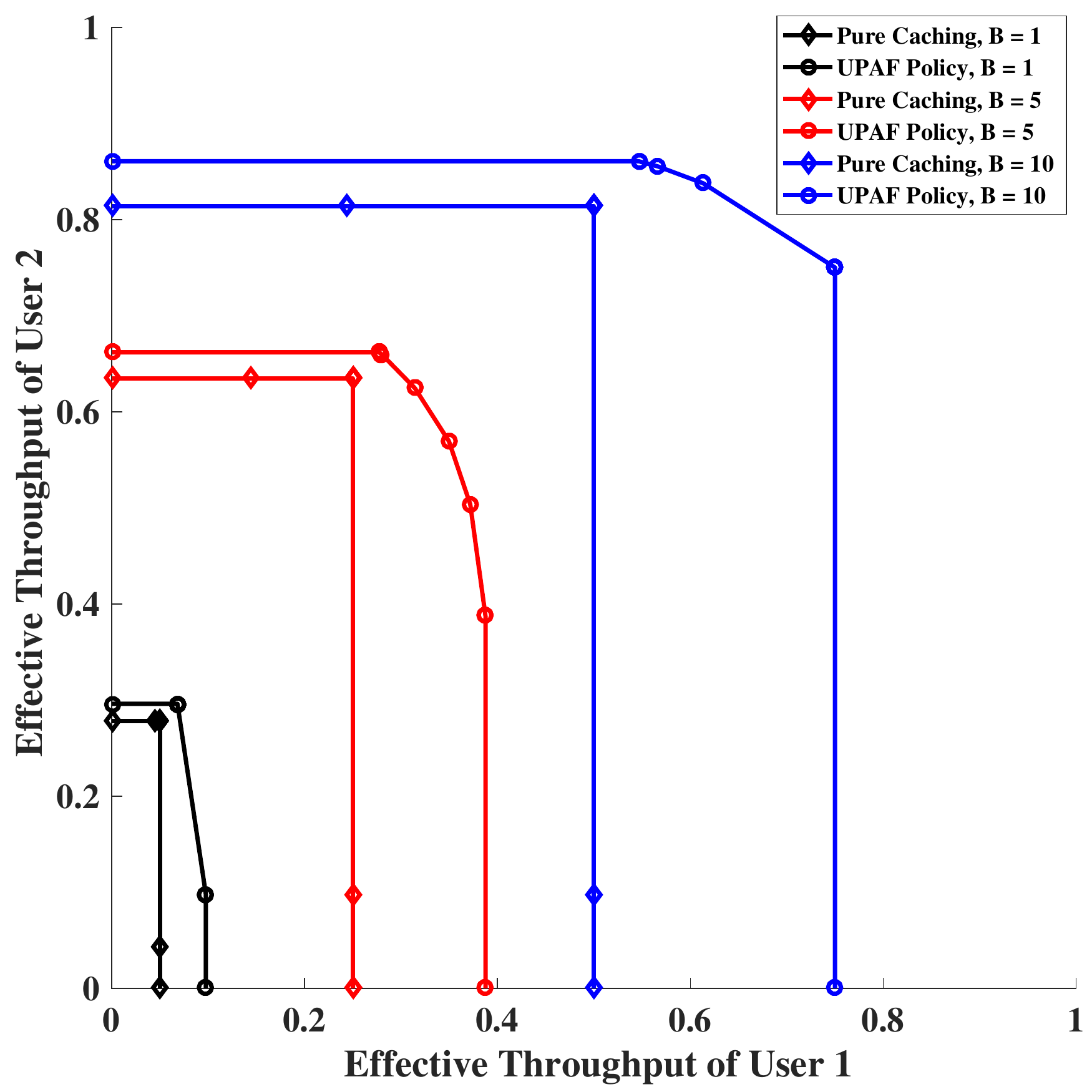}
    \caption{Achievable domain under UPAF policies for $(\bm{B},N,\mathcal{P})$-Caching with two users and different buffer sizes. Preference matrix $\bm{P}=\bm{P}_2.$}    
\label{different_b}  
    \end{minipage}
\end{figure}

In Fig. \ref{different_p}, we consider $(\bm{B},N,\mathcal{P})$-Caching with two users and present the achievable domain under UPAF policies for  three different preference matrices.
The number of content items and the buffer size are set to be $N=20$ and $B=1$, respectively. Two types of user preferences are considered, i.e.,  Zipf's law  $\bm{p}_{\text{zipf}} = \left[\frac{1^{-1}}{\sum_{i=1}^{20}i^{-1}},...,\frac{20^{-1}}{\sum_{i=1}^{20}i^{-1}}\right]$ and uniform distribution $\bm{p}_{\text{unif}} = \left[\frac{1}{20},...,\frac{1}{20}\right]$.  The three preference matrices are set to be
\begin{equation*}
    \begin{split}
    \bm{P}_1= \begin{bmatrix}
    \bm{p}_{\text{zipf}} \\
    \bm{p}_{\text{zipf}}
     \end{bmatrix},
          \bm{P}_2 = \begin{bmatrix}
    \bm{p}_{\text{unif}}\\
     \bm{p}_{\text{zipf}} 
      \end{bmatrix},
          \bm{P}_3 = \begin{bmatrix}
      \bm{p}_{\text{unif}}  \\
      \bm{p}_{\text{unif}} 
    \end{bmatrix}.
    \end{split}
\end{equation*}
  It is seen that UPAF policies always attain larger achievable domains than pure caching policies do, which demonstrates the potential of improving the effective throughput by user cooperation. Corresponding to the analysis in Section \ref{twousersection}, the  achievable domains under UPAF policies are polygons. When $\bm{P} = \bm{P}_1 $ or $\bm{P}_3$, the achievable domains have reflectional symmetry. This is because user $1$ and user $2$ have the same preference distribution. When $\bm{P} = \bm{P}_2 $, the preference of  user $2$ is more concentrated. In this case, the achievable domain under $\bm{P}_2$ loses the  reflectional symmetry and user $2$ can obtain  higher maximum effective throughput than user $1$ does.


In Fig. \ref{different_b}, we present the achievable domain under UPAF policies for $(\bm{B},N,\mathcal{P})$-Caching with two users and different buffer sizes. The preference matrix is set to be $\bm{P}=\bm{P}_2$. It is not surprising that the larger the buffer size, the larger the achievable domain. Again,  the achievable domains under UPAF policies are polygons and user 2 obtain higher maximum effective throughput than user 1 does. 

\begin{figure}
    \begin{minipage}[t]{0.5\linewidth}
    \centering
\includegraphics[width=2.7in]{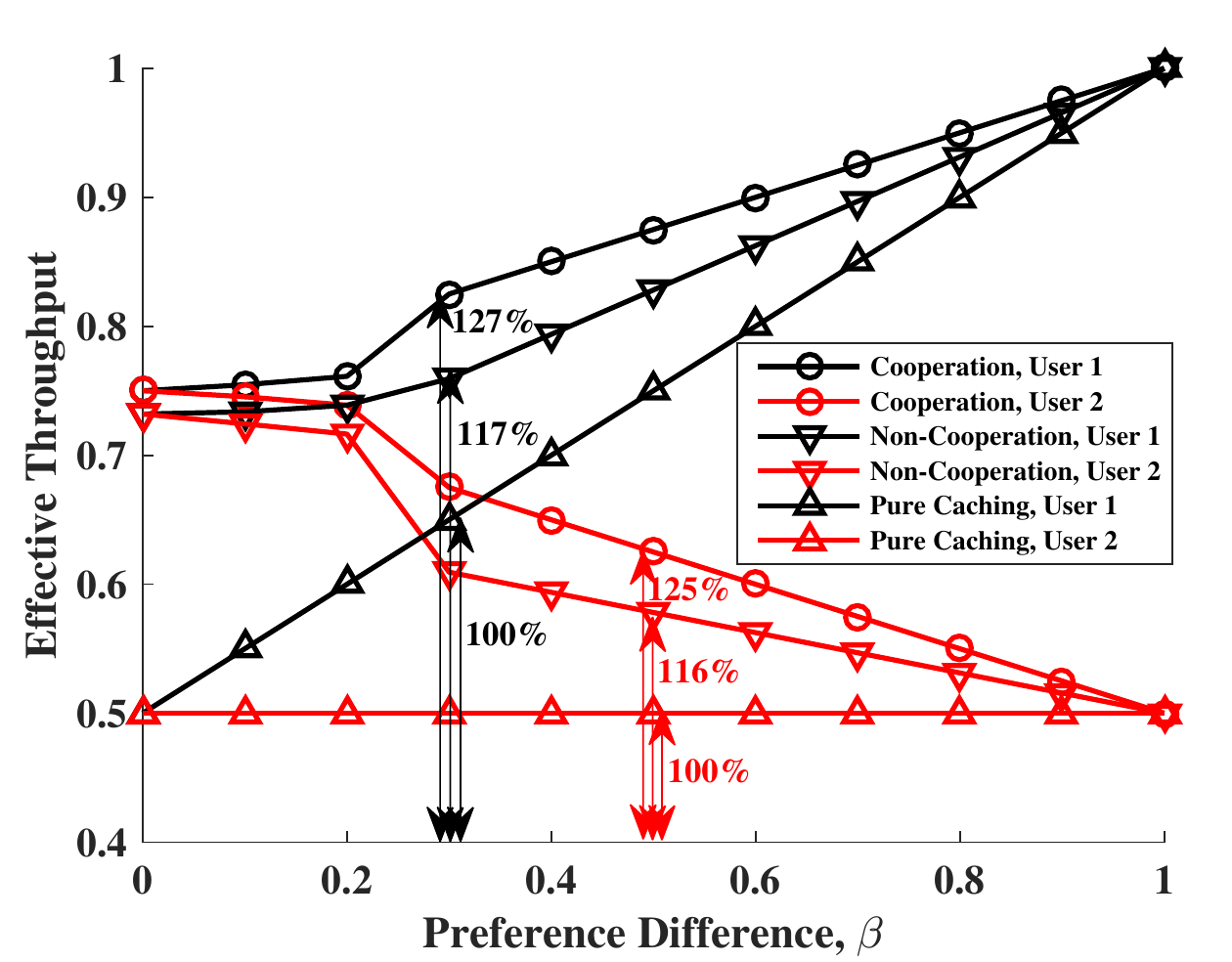}
\caption{Effective throughput versus preference difference $\beta$ for$\hspace{0.2cm}$ $(\bm{B},N,\mathcal{P})$-Caching with two users.}    
\label{game_beta}  
    \end{minipage}%
    \begin{minipage}[t]{0.5\linewidth}
    \centering
\includegraphics[width=2.7in]{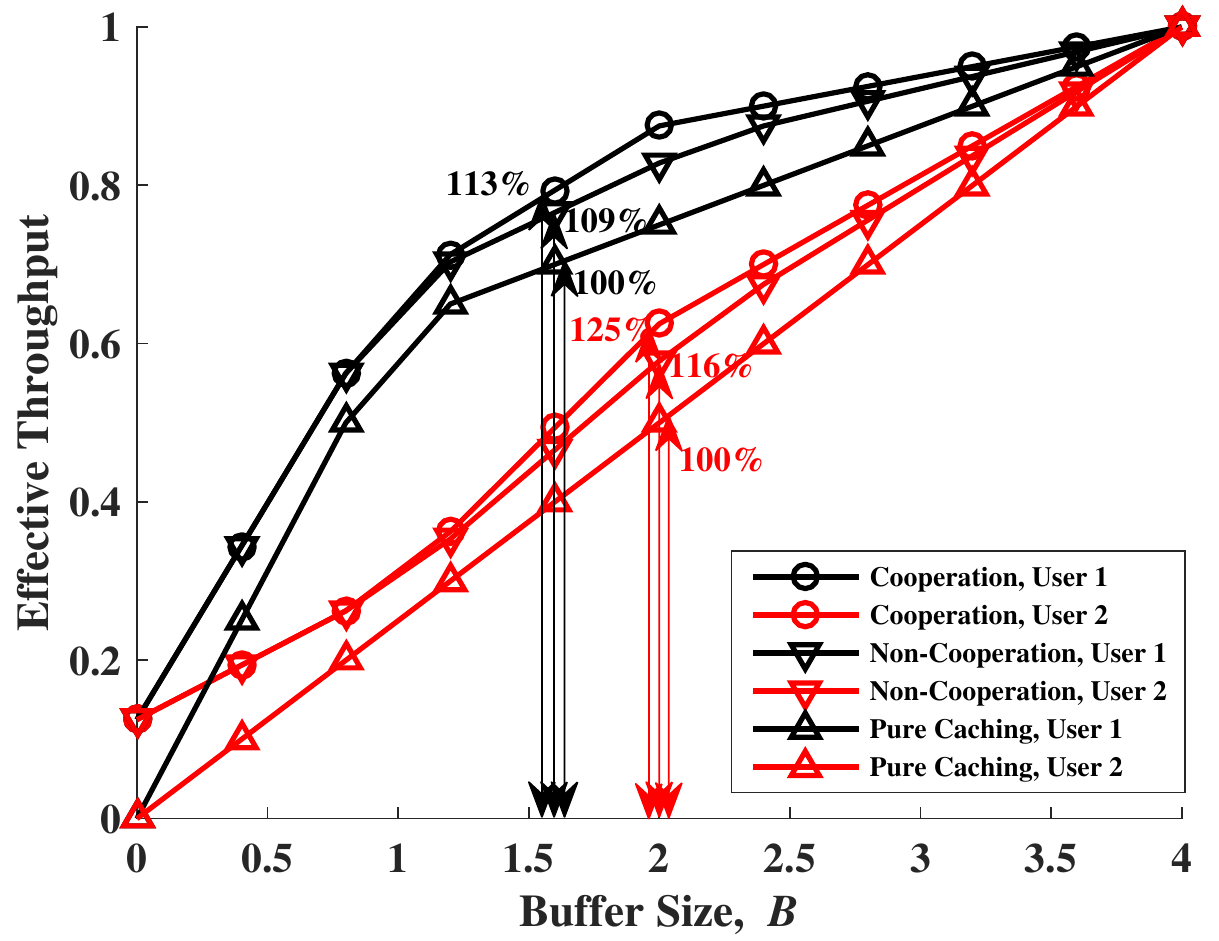}
\caption{Effective throughput versus buffer size for $(\bm{B},N,\mathcal{P})$-Caching with two users.}    
\label{game_buffer}  
    \end{minipage}
\end{figure}


To illustrate the advantage of user cooperation, we consider a $(\bm{B},N,\mathcal{P})$-Caching with two users, where the number of content items and the buffer size are set to be $N=4$ and $B=2$. The  preference matrix is given by 
\begin{equation*}
\begin{split}
\bm{P}(\beta)& = \begin{bmatrix}
0.25+ 0.75\beta  & 0.25(1-\beta) & 0.25(1-\beta)& 0.25(1-\beta) \\
0.25 & 0.25 & 0.25 & 0.25 
\end{bmatrix} .
\end{split}
\end{equation*}
The user preference  of user 1 is a linear combination of two vectors, $[0.25,0.25,0.25,0.25]$ and $[1,0,0,0]$. Thus the parameter $\beta$ describes the preference difference between the two users and how concentrated the user 1's preference is. Fig. \ref{game_beta} presents the effective throughput versus $\beta$.   It can be seen that cooperative games always reach higher effective throughputs than noncooperative games and pure caching do. The  effective throughput of user 1 increases with $\beta$, which validates again that more concentrated preferences helps to create higher effective throughputs. Under the setting in the simulation, user cooperation can even bring over 25\% additional effective throughput.   In contrast,  effective throughput of user 2 decreases with $\beta$. That is because concentrated preference helps user 1 to dominate the  game between the two users. Theorem \ref{psne} ensures that there exists a PSNE when $\beta$ is small. In the simulations, Algorithm \ref{algo1} always finds a PSNE  for $\beta$ varying from 0 to 1  ($T$ and $\varepsilon$ are set to be $100$ and $10^{-5}$ respectively). We also tested Algorithm \ref{algo1} for different $\bm{B}, N,$ and randomly generated $\bm{P}$. Algorithm \ref{algo1} returned a PSNE  with probability exceeding 90\%. Simulations demonstrate that a PSNE exists on a scale greater than that Theorem \ref{psne} characterizes.
%

Fig. \ref{game_buffer} presents effective throughput versus buffer size for $(\bm{B},N,\mathcal{P})$-Caching with two users. The parameter $\beta$ in the preference matrix $\bm{P}(\beta)$ is set to be $\beta =0.5$. It is not surprising to see that the effective throughputs always increase with the buffer size. Again, user 1 dominates the  games and  always obtains a higher effective throughput than user 2 does. User cooperation can bring a considerable increase in effective throughput. When the buffer size is small, the gap between cooperative game and noncooperative game is small. That is because multicasting opportunities are few due to the lack of side information. When the buffer size tends to the number of content items, the gaps between cooperative game, noncooperative game, and pure caching shrink, because few data need to be transmitted in the delivery phase.

 Fig. \ref{multiuser}   presents effective throughput versus buffer size for $(\bm{B},N,\mathcal{P})$-Caching in the general multiuser case.  In the simulation, we consider there are three users and the preference matrix is 
 \begin{equation*}
 \begin{split}
 \bm{P}_4& = \begin{bmatrix}
 0.7 & 0.2 & 0.1 & 0  \\
  0.4 & 0.3 & 0.2 & 0.1\\
 0.25 & 0.25 & 0.25 & 0.25
 \end{bmatrix}.
 \end{split}
 \end{equation*}
 For both Algorithm \ref{algo3}
 and pure caching, the effective throughput increases with the buffer size.   Algorithm \ref{algo3} always achieves  a higher effective throughput than pure caching does. However, the gap between Algorithm \ref{algo3}  and pure caching shrinks with increasing the buffer size. That is because the  effective throughput gain of Algorithm \ref{algo3}  mainly comes from multicasting for small buffer size. In Fig. \ref{multiuser}, user $1$ always obtains a higher effective throughput than user $2$ does, and user $2$ always achieves a higher effective throughput than user $3$ does. The more concentrated the preference, the higher effective throughput the user can obtain. 

   \begin{figure} 
    \centering
    \includegraphics[width=2.2in]{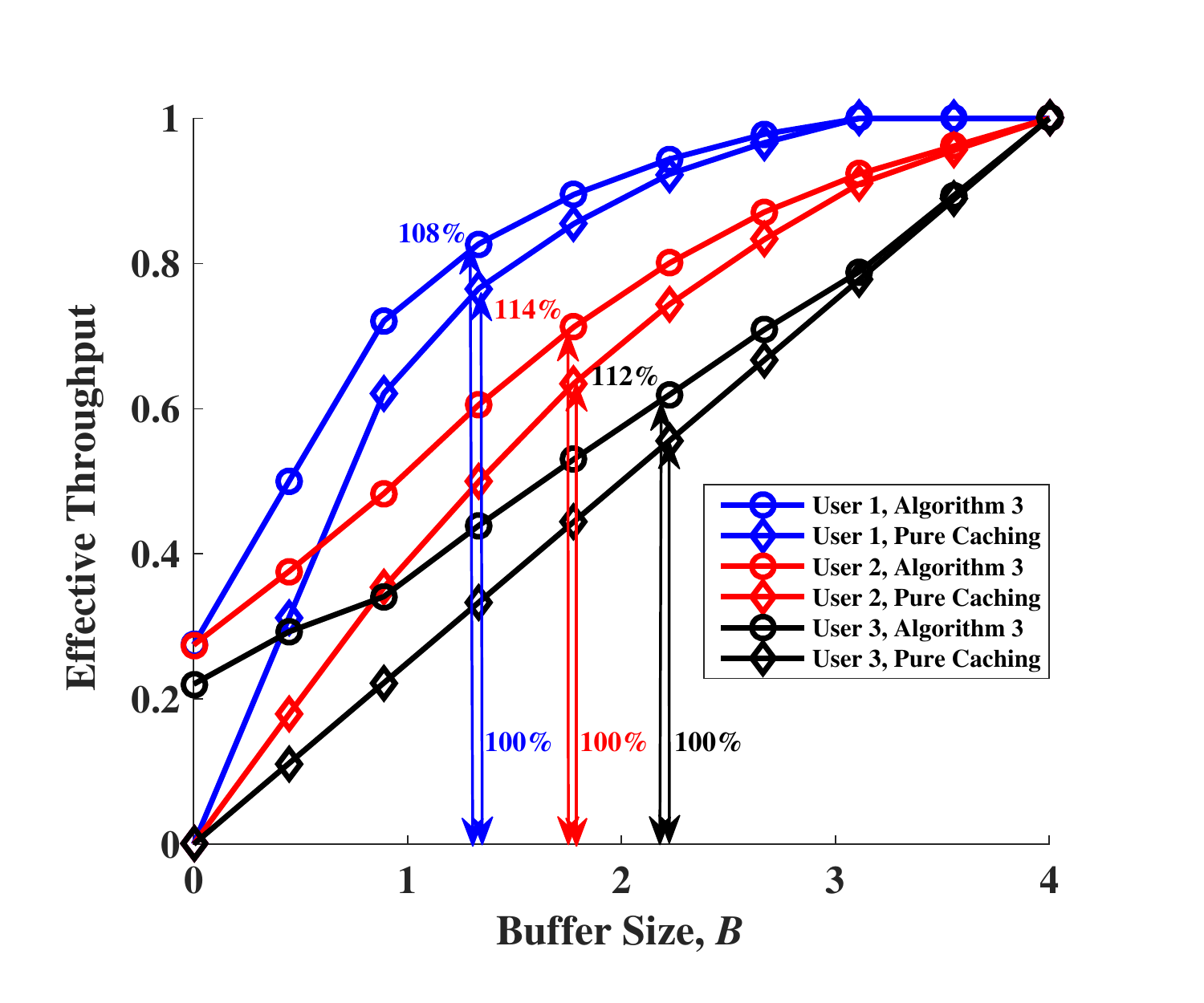}
    \caption{Effective throughput versus buffer size for $(\bm{B},N,\mathcal{P})$-Caching with three users.}    
    \label{multiuser}  
\end{figure}
\section{Conclusion}
In this paper, we have studied the impact of heterogeneous user preferences on effective throughput in cache-aided networks. In particular, the effective throughput is defined by splitting the transmission costs and can be viewed as the number of equivalent cache hits in each user's buffer. It has been proved that the achievable domain of effective throughputs   is a convex set and can be totally characterized by its boundary in the positive orthant.   For $(\bm{B},N,\mathcal{P})$-Caching with two users, both theoretical analysis and simulation results have demonstrated  that the achievable domain 
under UPAF policies is a polygon. A noncooperative game has  been formulated to investigate Nash equilibria between the two users. It has been proved that the noncooperative game has  PSNEs if user preferences are similar. A low-complexity algorithm has been proposed to provide a PSNE for the noncooperative game with high probability. In addition, an effective throughput allocation scheme has been given based on  a cooperative game.  For $(\bm{B},N,\mathcal{P})$-Caching in the general multiuser case, a feasible UPAF policy has also been proposed to organize user cooperation. It has  also been shown that user cooperation can bring significant caching gains and the achievable domain of effective throughputs   expands with the increase of the buffer size. Users with more concentrated preferences can obtain  higher effective throughputs.
Significant future topics include games for $(\bm{B},N,\mathcal{P})$-Caching with multiple users, practical user preference analysis based on real data, and the analysis of effective throughput in cache-aided D2D networks and heterogeneous networks.
 
\bibliographystyle{IEEEtran}
\linespread{0.86}

\end{document}